\documentclass[a4paper, reqno, 11pt]{amsart}        

\usepackage{amssymb}

\usepackage{epsfig}  		
\usepackage{epic,eepic}       

\usepackage{enumitem}

\let\oldtocsection=\tocsection
 
\let\oldtocsubsection=\tocsubsection
 
\let\oldtocsubsubsection=\tocsubsubsection
 
\renewcommand{\tocsection}[2]{\hspace{0em}\oldtocsection{#1}{#2}}
\renewcommand{\tocsubsection}[2]{\hspace{1em}\oldtocsubsection{#1}{#2}}
\renewcommand{\tocsubsubsection}[2]{\hspace{2em}\oldtocsubsubsection{#1}{#2}}

\usepackage{amsfonts}
\usepackage{amsthm}
\usepackage{amsmath}
\usepackage{amscd}
\usepackage[latin2]{inputenc}
\usepackage{t1enc}
\usepackage[mathscr]{eucal}
\usepackage{indentfirst}
\usepackage{graphicx}
\usepackage{graphics}
\usepackage{pict2e}
\usepackage{epic}
\numberwithin{equation}{section}
\usepackage[margin=2.9cm]{geometry}
\usepackage{hyperref}
\usepackage{dsfont}
\usepackage{csquotes}
\usepackage{stmaryrd}
\usepackage{blindtext}
\usepackage{tikz-cd}
\usepackage{mathrsfs}
\usepackage{cite}
\usepackage{mathtools}
\usepackage{extarrows}
\usepackage{epstopdf}
\usepackage[T1]{fontenc}
\usepackage{braket}

\theoremstyle{definition}
\newtheorem{definition}[equation]{Definition}
\newtheorem{example}[equation]{Example}
\newtheorem{proposition}[equation]{Proposition}
\newtheorem{theorem}[equation]{Theorem}
\newtheorem{remark}[equation]{Remark}

\newtheorem{corollary}[equation]{Corollary}
\newtheorem{lemma}[equation]{Lemma}

\numberwithin{equation}{section}

\newcommand{\midwedge}{\text{\Large$\wedge$}}
\newcommand{\midodot}{\text{\Large$\odot$}}

\newcommand{\be}{\begin{equation}}
\newcommand{\ee}{\end{equation}}
\def\beqa{\begin{eqnarray}}
\def\eeqa{\end{eqnarray}}
\def\bean{\begin{eqnarray*}}
\def\eean{\end{eqnarray*}}

\newcommand{\de}{\mathrm{d}}

\newcommand{\del}{\partial}

\newcommand{\De}{\mathrm{D}}
\newcommand{\gv}{\mathtt{GV}}
\newcommand{\gvs}{\mathtt{gv}}
\newcommand{\ttG}{\mathtt{G}}
\newcommand{\ttC}{\mathtt{C}}
\newcommand{\ttA}{\mathtt{A}}

\newcommand{\IR}{\mathbb{R}}

\newcommand{\frg}{\mathfrak{g}}

\def\e{{\,\mathrm e}\,}

\newcommand{\cH}{{\mathcal H}}

\newcommand{\cF}{{\mathcal F}}

\newcommand{\sfA}{{\mathsf{A}}}

\newcommand{\sfS}{{\mathsf{S}}}

\newcommand{\sfH}{{\mathsf{H}}}
\newcommand{\sfG}{{\mathsf{G}}}

\newcommand{\sfO}{{\mathsf{O}}}

\newcommand{\sfF}{{\mathsf{F}}}

\newcommand{\unit}{\mathds{1}}   			

\newenvironment{myitemize}{\begin{itemize}[itemsep=3pt, leftmargin=*, topsep=0.1cm]}{\end{itemize}}

\mathtoolsset{showonlyrefs}

\begin{document}

\title[V.~E.~Marotta and R.~J.~Szabo]{Godbillon-Vey 
Invariants of Non-Lorentzian Spacetimes \\[5pt] and Aristotelian Hydrodynamics}

\author[V.~E.~Marotta]{Vincenzo Emilio Marotta}
\address[Vincenzo Emilio Marotta]
{Mathematical Institute, Faculty of Mathematics and Physics\\ Charles University\\ Prague 186 75\\ Czech Republic}
\email{marotta@karlin.mff.cuni.cz}

\author[R.~J.~Szabo]{Richard J.~Szabo}
  \address[Richard J.~Szabo]
  {Department of Mathematics, Maxwell Institute for Mathematical Sciences and The Higgs Centre for Theoretical Physics\\
  Heriot-Watt University\\
  Edinburgh EH14 4AS \\
  United Kingdom}
\email{R.J.Szabo@hw.ac.uk}

\vfill

\begin{flushright}
\footnotesize
{\sf EMPG--23--03}
\normalsize
\end{flushright}

\vspace{1cm}

\begin{abstract}
We study the geometry of foliated non-Lorentzian spacetimes in terms of the Godbillon-Vey class of the foliation. We relate the intrinsic torsion of a foliated Aristotelian manifold to its Godbillon-Vey class, and interpret it as a measure of the local spin of the spatial leaves in the time direction. With this characterisation, the Godbillon-Vey class is an obstruction to integrability of the $\sfG$-structure defining the Aristotelian spacetime. We use these notions to formulate a new geometric approach to hydrodynamics of fluid flows by endowing them with Aristotelian structures. We establish conditions under which the Godbillon-Vey class represents an obstruction to steady flow of the fluid and prove new conservation laws. 
\end{abstract}

\maketitle

{
\tableofcontents
}


\section{Introduction}

The geometrisation of Newtonian gravity, known as Galilean or Newton-Cartan geometry, has as its local symmetries the Galilean symmetries of non-relativistic physics. Physics in the opposite ultra-local regime is geometrized by Carrollian geometry. Merging the two notions together gives Aristotelian geometry without any local Galilean or Carrollian boost symmetry, which is the  main focus of the present paper. These non-Lorentzian versions of spacetime geometry are of interest both as approximations of underlying relativistic theories and as interesting theories in their own right which have lately been under intensive investigations. For recent reviews we refer to~\cite{Bergshoeff:2022eog} for applications to non-Lorentzian particle dynamics and field theory, and to~\cite{Oling:2022fft} for applications to non-relativistic string theory.

On Galilean and Aristotelian manifolds one typically wishes to locally distinguish a special direction which is associated to absolute time, i.e. to give a local notion of simultaneous events and causality. This requires that the non-Lorentzian manifold admits a codimension one foliation, i.e. that it is {integrable}. The leaves of the foliation are then interpreted as the spatial slices of the spacetime. In this paper we explore the meanings of invariants of the foliation in terms of the geometry and physics of the non-Lorentzian spacetime, focusing on the Godbillon-Vey class~\cite{Godbillon1971} which is a well-known classical invariant in the foliation theory and dynamical systems literatures.

The Godbillon-Vey class $\gv(\cF)$ of a codimension one foliation $\cF$ of an $n$-dimensional manifold $M^n$ plays a crucial role in the study of the topology and dynamics of foliations, see e.g.~\cite{Reinhart1973,Hurder2018}.  It is a degree three de~Rham cohomology class which is a foliation-invariant: it is invariant under diffeomorphisms and foliated concordance; in three dimensions it is also a cobordism invariant. When $\gv(\cF)\neq 0$, there are leaves of $\cF$ with exponential growth. In~\cite{Brooks1984} the Godbillon-Vey class for a foliated oriented three-manifold $M^3$ is interpreted as a topological volume density that measures the complexity of representing the fundamental class $[M^3]$ by singular simplices; when $M^3$ is a hyperbolic three-manifold, with a transversely projective foliation $\cF$ defined by a monodromy representation of the fundamental group $\pi_1(M^3)\to{\sf PSL}(2,\IR)$, the integrated Godbillon-Vey invariant, i.e. the Godbillon-Vey number $\gvs(\cF)$, coincides with the geometric volume of $M^3$.

The Godbillon-Vey invariant measures a sort of ``twisting'' of the leaves of $\cF$: the geometric interpretation of $\gv(\cF)$ in three dimensions is due to Thurston~\cite{Thurston1972}, who describes it as a measure of the ``helical wobble'' of the foliation $\cF$ and constructs foliations with arbitrary Godbillon-Vey numbers. An explicit realisation of Thurston's helical wobble interpretation is given by~\cite{Reinhart1973} wherein the Godbillon-Vey invariant of a foliated Riemannian three-manifold is expressed in terms of the curvatures of the leaves and their normal bundles. 

We apply these notions to the codimension one foliation $\cF$ of an integrable Aristotelian manifold $M^n$, and pursue in detail the role that these two threads together play in the natural setting of \emph{hydrodynamics}. Hydrodynamics is an effective field theory which provides a universal description of a broad class of physical phenomena near thermal equilibrium in the long wavelength limit. Its equations of motion express the conservation of currents, which are parametrized in terms of fluid variables such as fluid velocity and density. The role of non-Lorentzian geometry in fluid mechanics becomes evident when one recalls that the classical Navier-Stokes equations are derived from Newton's laws as a description of the velocity $\boldsymbol v$ and pressure $p$ of a fluid in time and space. 

Written in vector calculus notation, the incompressible Navier-Stokes equations for a viscous fluid flowing in a simply connected domain in $\IR^3$ consist of a time evolution equation and the equation for divergence-free flow:
\be \label{eq:IntroEuler}
\frac{\partial \boldsymbol v}{\partial t} + \frac{1}{2}\,\boldsymbol\nabla |\boldsymbol v|^2 = \boldsymbol v\times\boldsymbol\xi - {\boldsymbol\nabla} p + \nu \, \boldsymbol\nabla^2\boldsymbol v \qquad \mbox{and} \qquad \boldsymbol\nabla\cdot\boldsymbol v = 0 \ ,
\ee
where $\nu$ is the viscocity. The evolution of the fluid vorticity $\boldsymbol\xi=\boldsymbol\nabla\times\boldsymbol v$ is governed by
\be \label{eq:IntroEH}
\frac{\partial\boldsymbol\xi}{\partial t} + \boldsymbol\nabla_{\!\boldsymbol v}\,\boldsymbol \xi = \boldsymbol\nabla_{\!\boldsymbol\xi}\,\boldsymbol v + \nu\,\boldsymbol\nabla^2\boldsymbol\xi \ ,
\ee
where $\boldsymbol\nabla_{\!\boldsymbol v}=\boldsymbol v\cdot\boldsymbol\nabla$ denotes the directional derivative along $\boldsymbol v$.
In this paper we consider only ideal fluid flows, which are described by dropping the diffusion terms, i.e. $\nu=0$; in this case Equations \eqref{eq:IntroEuler} reduce to the incompressible Euler equations and Equation \eqref{eq:IntroEH} to the Euler-Helmholtz equation.

Recent developments have led to generalisations of these non-relativistic fluid flows to curved spacetimes as well as to non-Lorentzian settings without Galilean boost-invariance. Fluid flows on Newton-Cartan geometries are discussed in e.g.~\cite{Geracie:2015xfa,Poovuttikul:2019ckt,Armas:2019gnb}, while Carrollian boost-invariant fluids are treated in e.g.~\cite{deBoer:2017ing,Poovuttikul:2019ckt,Freidel2022,Bagchi:2023ysc}. As discussed in~\cite{deBoer:2020xlc}, the breaking of boost symmetry is natural for many systems, particularly once a preferred observer reference frame is fixed, and the relevant curved background geometry is then Aristotelian spacetime; see also~\cite{Armas:2020mpr,Armas:2023ouk,Jain:2023nbf}.
The hydrodynamic equations governing non-Lorentzian fluid flows are usually derived as limits of the relativistic conservation laws of general relativistic fluids~\cite{Ciambelli:2018xat, Petkou:2022bmz}. In the
case of Newtonian fluids, i.e. non-relativistic hydrodynamics, one thereby obtains the
Navier-Stokes equations together with the  conservation law for the fluid density. 

The relevance of the Godbillon-Vey class in hydrodynamics has been noted in many instances, see e.g.~\cite{Tur1993,Arnold,Webb2014,Webb2019,Machon2020a,Machon2020}; it plays a prominent role in any hydrodynamic system described by a one-form which is conserved by the fluid. In ideal fluid dynamics, advected topological invariants, i.e. invariants that are conserved in the comoving reference frame to the flow, and their conservation laws have a wide range of physical applications. These include the vorticity, as well as the hydrodynamic helicity which measures the topology of vorticity fields. For integrable fluid flows, whose vorticity vector field is tangent to the leaves of a codimension one foliation $\cF$ of the fluid domain, the helicity vanishes but the vorticity can still have non-trivial topology if its Godbillon-Vey invariant is non-zero~\cite{Tur1993,Webb2014,Machon2020}. Then the Godbillon-Vey number $\gvs(\cF)$ characterises the topological type of the two-dimensional foliation defined by the fluid vorticity. 

In this paper we will extend these local considerations on open domains in $\IR^3$ to arbitrary curved backgrounds for fluids without any boost symmetry. We introduce a novel hydrodynamic theory of `Aristotelian fluids': An Aristotelian fluid is a fluid which flows on an Aristotelian manifold. Our aim is to reformulate and unify earlier results in the framework of Aristotelian geometry. In particular, we generalize approaches to ideal integrable flows to fluid domains which are arbitrary oriented Aristotelian manifolds of any dimension. In this way, Aristotelian geometry aids in characterising physical features of known solutions of the Euler equations, and in constructing new ones. 

One drawback in our approach is that we only work with smooth foliations, which excludes the interesting fixed points of the flow equations for fluid lines where interesting changes in topology of a fluid can occur. Incorporating fixed points generally requires the use of singular foliations which, although possible in principle, are technically very difficult to work with and are currently out of reach with our methods. It would be interesting to extend our techniques to include these cases.

Let us mention another potential (albeit speculative) application of our framework to the physics of fractons, which are quasi-particle condensed matter configurations with restricted mobility that display UV/IR mixing and subsystem symmetries. They can be described in terms of foliated field theories built by coupling together fields supported on the leaves of foliations of spacetime, see e.g.~\cite{Geng:2021cmq} for a review. Some theories of fractons on curved spacetimes couple to Aristotelian geometries with general intrinsic torsion~\cite{Bidussi:2021nmp,Jain:2021ibh}. The Godbillon-Vey class of an integrable Aristotelian structure in this setting describes the local spinning motion of fractons in their restricted mobility directions. Some fracton superfluids have been recently described as ideal Aristotelian superfluids in~\cite{Armas:2023ouk,Jain:2023nbf}. We leave the problem of addressing the implications of our formalism in this context for future work.

\medskip

\subsection{Summary of Results} ~\\[5pt]
Let us now summarise the main results of this paper in more detail.

In this paper we discuss the interplay between the geometry of a non-Lorentzian spacetime endowed with a spatial foliation of codimension one and the fundamental tensors characterising it such as its intrinsic torsion. We show how the intrinsic torsion of a spacetime structure determines the characteristic class of the spatial foliation, the  \emph{Godbillon-Vey class}. 
We find that a representative of the Godbillon-Vey class of a foliated non-Lorentzian spacetime is completely determined for an Aristotelian structure. In this case its main constituents are the intrinsic torsion and the vector field of observers. We show that a non-vanishing Godbillon-Vey class for a foliated Aristotelian manifold is an obstruction to integrability of its underlying $\sfG$-structure,\footnote{Beware that there are two generally unrelated notions of `integrability' that permeate this paper: Frobenius integrablity of a distribution in $TM^n$ which decides when it yields a foliation of $M^n$, and integrability of a $\sfG$-structure on $M^n$ which decides when there are local frames of $TM^n$ with standard flat $\sfG$-structures.} which measures the local spin of each spatial leaf in the time direction, i.e. along the integral curves of the vector field of observers. In particular, torsion-free Aristotelian spacetimes always have trivial Godbillon-Vey class. We suggest that the Godbillon-Vey invariant adds a finer topological graining to the classification of non-Lorentzian spacetime structures given in~\cite{Figueroa2020}.

Our main application of the relation between the Godbillon-Vey class and a spacetime structure discussed in this paper is to `Aristotelian hydrodynamics'. 
We give a precise definition for our notion of an `Aristotelian fluid'. In our picture, an Aristotelian fluid is characterised by an $n$-dimensional manifold $M^n$ endowed with a quadruple $(\tau, v, \mu,g)$ of a one-parameter family of clock forms $\tau$, the fluid velocity field $v$, the fluid density $\mu$, and the Riemannian metric $g$.
The distribution $\ker(\tau)$ determines a family of foliations that are orthogonal to the integral curves of $v,$ i.e.~the fluid lines, with respect to the background metric $g.$ In this context we thus interpret the Aristotelian structure differently from its canonical applications in the description of non-Lorentzian spacetimes: the clock form $\tau$ here describes the direction along which the fluid flows, rather than the time, and the orthogonality condition on the family of foliations with the fluid flow determines a family of spatial metrics on the leaves.   

The hydrodynamic equations comprise the usual conservation laws as well as the Euler equations for the pressure. 
Similarly to~\cite{deBoer:2020xlc}, a key role in our approach is played by the square of the fluid speed, denoted $s_v$ in the following ($s_v=|\boldsymbol v|^2$ in Equation~\ref{eq:IntroEuler}); its thermodynamic dual is the kinematic mass density introduced in~\cite{deBoer:2017ing} for the hydrodynamic description of non-boost invariant systems. We determine the transport equations for the speed of the fluid and the clock form. In particular, we generalize results of Machon~\cite{Machon2020}, not only to arbitrary integrable Aristotelian manifolds, but also to fluids with non-constant speed along the fluid lines.
For ideal Galilean fluids, the fluid density $\mu$ defines a conserved spin zero quantity, together with the continuity equation for the vorticity and the Euler equations. For our ideal Aristotelian fluids with constant speed, there is an additional conserved spin one current given by the clock form $\tau$ of the Aristotelian structure.

We demonstrate how the torsion tensor of a foliated Aristotelian fluid is completely determined by its speed and vorticity. Then we show that this construction is compatible with the classification of non-Lorentzian spacetime structures presented in \cite{Figueroa2020}. We determine circumstances under which the torsion of the Aristotelian fluid vanishes, so that there is an $\sfS\sfO(n-1)$-frame moving along the fluid flow (at least at first order). We also show that the transport equation for the torsion tensor is mainly ruled by the speed of the fluid and the pressure field.

We provide a complete characterisation of the Godbillon-Vey class of an Aristotelian fluid flow in terms of its vorticity, speed and density. If the speed of the fluid is constant along the integral curves of the velocity field, i.e. the vector field of observers of the Aristotelian structure, then the Godbillon-Vey class represents an obstruction to steady flow, i.e. to the existence of a stationary solution of the Euler equations, corresponding to an equilibrium state of the fluid flow. This allows us to analyse circumstances under which Aristotelian fluids admit a steady flow with a Bernoulli function. We also show that the Godbillon-Vey class is transported exactly by the fluid flow.
As an example of fluid flow with trivial Godbillon-Vey class, we consider $n$-dimensional warped products and describe their family of Aristotelian structures together with an example of a stationary solution.

In two dimensions, we show how essentially any ideal fluid is naturally described as an Aristotelian fluid. In particular, we characterise the Aristotelian structure in terms of the stream function of the fluid. We exhibit some classical examples of two-dimensional fluid flows in which there is a relation between points where the torsion vanishes and singular regions which may be viewed as sources of vorticity.

In three dimensions, we determine the flow equation for the Godbillon-Vey number and show that it is a conserved quantity which measures of the spin of the leaves of the foliation in the direction of the fluid lines. 
In this case, we observe that the triviality of the Godbillon-Vey class for an unsteady fluid flow implies that the velocity field preserves the direction of the vorticity vector field. We also exhibit an example of a three-dimensional Aristotelian fluid flow with non-trivial Godbillon-Vey class built on the Roussarie foliation of ${\sf PSL}(2, \IR)$ and its quotient by a torsion-free cocompact discrete subgroup.

\medskip

\subsection{Outline} ~\\[5pt]
This paper is organised as follows.

In Section \ref{sect:nonrel} we briefly review the geometry of non-Lorentzian spacetimes. We discuss Galilean structures (also known as Newton-Cartan geometries), highlighting some important topological aspects for foliated spacetimes. We further describe Carrollian structures and Aristotelian structures, again focusing on the geometric properties of the foliated case. 

In Section \ref{sect:gvclass} we first recall the definition and main properties of the Godbillon-Vey class of a codimension one foliated manifold, and discuss its interpretation by presenting the  \emph{helical wobble}, originally introduced by Thurston.
After briefly reviewing the properties of the intrinsic torsion of a non-Lorentzian structure, we show how the torsion characterises the Godbillon-Vey class of a non-Lorentzian spacetime, and in particular for an Aristotelian structure. 

Section \ref{sect:IdealAri} is devoted to the study of ideal incompressible fluid flows endowed with an Aristotelian structure admitting a codimension one foliation. After reviewing the basics of ideal hydrodynamics, we proceed by defining incompressible fluids endowed with a family of Aristotelian structures.  This allows us to describe the time evolution of these structures by determining their transport equations. We completely determine the torsion and Godbillon-Vey class in terms of fluid variables, and discuss their specific properties in the special cases of fluid flows in two and three dimensions, together with concrete examples.

\medskip

\subsection{Acknowledgements} ~\\[5pt]
V.E.M. thanks Leo Bidussi for very helpful discussions. R.J.S. thanks the Mainz Institute for Theoretical Physics (MITP) of the Cluster of Excellence PRISMA${}^+$ (Project ID 39083149), the Centro de
Matem\'atica, Computa\c{c}\~ao e Cogni\c{c}\~ao of the Universidade de Federal do ABC (S\~ao Paulo,
Brazil), and the Institute for Physics of Humboldt University Berlin for hospitality and support during various stages of this work. The work of {\sc V.E.M.} was supported in part by a
Maxwell Institute Research Fellowship and by the GACR Grant EXPRO 19-28268X. The work of {\sc R.J.S.} was supported in part by
the STFC Consolidated Grant ST/P000363/1 and by the FAPESP Grant 2021/09313-8.

\section{Non-Lorentzian Spacetimes} \label{sect:nonrel}

In this section we briefly recall the relevant spacetime structures that will play a role in this paper, following~\cite{Figueroa2020}. For a complete description of kinematics and dynamics on non-Lorentzian spacetimes, we refer to the review~\cite{Bergshoeff:2022eog} and references therein.

\medskip

\subsection{Galilean Structures} ~\\[5pt]
Non-Lorentzian manifolds may be regarded as equipped with local causal structures that come from taking different limits of the speed of light $c$ in a local inertial reference frame on a Lorentzian manifold. In this sense, Galilean structures arise from the limit $c\to\infty$, i.e. the non-relativistic limit. They capture the kinematics of the Newtonian counterpart of relativistic structures, and are often refered to as Newton-Cartan geometries. For a complete characterisation of Galilean structures, see~\cite{Bekaert:2014bwa} and references therein. 

\begin{definition}
A \emph{Galilean structure} on an $n$-dimensional manifold $M^n$ is a pair $(\tau, \gamma)$ of a nowhere-vanishing one-form $\tau \in \mathsf{\Omega}^1(M^n),$ called the \emph{clock form}, and a corank one positive-semidefinite symmetric tensor\footnote{The symbol $\odot^\bullet$ denotes the symmetric tensor product. Throughout we use the symbol $\iota_{(\cdot)}$ to denote interior product of a tensor with a one-form or a vector field.} $\gamma \in \mathsf{\Gamma}(\midodot^2\, TM^n)$ such that $\iota_\tau \gamma = 0,$ called the \emph{spatial cometric}. A \emph{Galilean manifold} (or \emph{Galilean spacetime}) is a manifold endowed with a Galilean structure.
\end{definition}

Since $\ker(\gamma) = \mathrm{Span}(\tau) \subset T^*M^n,$  the spatial cometric $\gamma$ is a metric tensor on the subbundle $\ker(\tau)^* \subset T^*M^n.$

\begin{remark} \label{rmk:GG}
To characterise the $\sfG$-structure defining a Galilean manifold, consider the vector space $V=\IR^n$ with choice of basis $(H, s_1, \dots , s_{n-1})$ and dual vector space $V^*$ with dual basis $(\eta, \sigma^1, \dots , \sigma^{n-1}).$ Then the structure group of a Galilean manifold is the subgroup $\sfG_\ttG \subset \mathsf{GL}(n,\IR)$ preserving $\eta $ and $\delta^{ab} \, s_a \odot s_b$:
\be 
\sfG_\ttG = \Big\{  
{\small \bigg(\begin{matrix}
1 & 0 \\
v & \sfA
\end{matrix}\bigg) } \normalsize
\ \Big| \  v \in \IR^{n-1} \ , \ \sfA \in \sfO(n-1)  \Big\} \ .
\ee
This is the usual inhomogeneous group of Galilean transformations, whose component connected to the identity consists of local spatial rotations and local Galilean boosts.
The group $\sfG_\ttG$ is isomorphic to the semi-direct product $\sfO(n-1) \ltimes \IR^{n-1},$ with Lie algebra
\be 
\frg_\ttG = \Big\{   {\small
\bigg(\begin{matrix}
0 & 0 \\
v & A
\end{matrix}\bigg) } \normalsize
\ \Big| \ v \in \IR^{n-1} \ , \ A \in \mathfrak{so}(n-1) \Big\} \ .
\ee
\end{remark}

For all of our spacetime structures we will assume the existence of a compatible linear connection. A compatible Galilean connection $\nabla^\ttG$ is a linear connection which preserves the structure tensors:
\be  
\nabla^\ttG \tau = 0 \qquad \mathrm{and} \qquad  \nabla^\ttG \gamma = 0 \ .
\ee
Its torsion $T^{\nabla^\ttG} \in \mathsf{\Gamma}(\midwedge^2\, T^*M^n \otimes TM^n)$ is defined as usual by
\begin{align}  
    T^{\nabla^\ttG}(X,Y) = \nabla^\ttG_XY - \nabla^\ttG_YX - [X,Y] \ ,
\end{align}
where $[X,Y]$ is the Lie bracket of vector fields $X,Y\in\mathsf{\Gamma}(TM^n)$, and it is easy to prove~\cite{Figueroa2020}

\begin{proposition} \label{prop:Ztau}
Let $(\tau, \gamma)$ be a Galilean structure on $M^n$ together with a Galilean connection $\nabla^\ttG$, and denote its torsion by $T^{\nabla^\ttG} \in \mathsf{\Gamma}(\midwedge^2\, T^*M^n \otimes TM^n).$ Then
\be \label{eq:detau}
\tau \circ T^{\nabla^\ttG} = \de \tau \ .
\ee 
\end{proposition}

Galilean structures are of infinite type \cite{Kunzle1972}. This superficially makes the construction of a classifying Lie algebroid for them elusive. It would be interesting to investigate generally which spacetime structures admit a classifying Lie algebroid.

\begin{remark}
The clock form $\tau$ defines a codimension one foliation of $M^n$ if it is \emph{integrable}, i.e. when
\be \label{eq:integrability}
\tau\wedge\de\tau = 0 \ .
\ee
This equation is equivalent to the Frobenius integrability condition for sections of the distribution $\ker (\tau) \subset TM^n$ of rank $n-1.$ The foliation determines the spatial leaves of the spacetime $M^n$, giving a notion of absolute simultaneity and Newtonian causality.
In particular, a spacetime structure whose clock form is closed always admits at least a local definition of Newtonian absolute time. Galilean structures with integrable clock form $\tau$ are called \emph{torsionless Newton-Cartan geometries} if $\de \tau = 0$, and \emph{twistless torsional Newton-Cartan geometries} when $\de \tau \neq 0$ and $\tau \wedge \de \tau = 0$ \cite{Figueroa2020}.

It is easy to show that Equation~\eqref{eq:integrability} implies
\be \label{eq:alphaint}
\de \tau = \tau \wedge \alpha \ ,
\ee
for some $\alpha \in \mathsf{\Omega}^1(M^n),$ see for instance \cite{Figueroa2020}. The one-form $\alpha$ is not uniquely determined. 
\end{remark}

\begin{example}[\textbf{Global Absolute Time}] \label{ex:globaltime}
Let us discuss an example of these spacetime structures that highlights their topology.
Let $M^n$ be a spacetime determined by a nowhere-vanishing exact one-form $\tau = \de f$ with $f \in C^\infty(M^n).$ Since $\tau$ is nowhere-vanishing, $f$ has no critical points, hence $f \colon M^n \rightarrow \IR$ is a submersion and the foliation is given by the fibres of this map. We can interpret $f$ as defining a global absolute time.
\end{example}

Example~\ref{ex:globaltime} hints at how the Global Reeb-Thurston Stability Theorem~\cite{Moerdijk2003} can be used in this context to describe the geometry of time.

\begin{theorem}[\textbf{Global Reeb-Thurston Stability}] \label{thm:Globalreeb}
Let $M^n$ be a compact connected $n$-manifold endowed with a codimension one transversely orientable foliation $\cF$ admitting a compact leaf $L_0$ with trivial cohomology in degree one. Then $M^n$ is a fibre bundle over the circle $S^1$ with fibres given by the leaves of $\cF$, all of which are diffeomorphic to $L_0.$ 
\end{theorem}

Theorem~\ref{thm:Globalreeb} provides conditions under which an integrable Galilean structure on $M^n$ admits a ``periodic time''. The condition of transverse orientability is always met because of the existence of the nowhere-vanishing one-form $\tau.$ 

\begin{corollary} \label{cor:Reeb}
Let $M^n$ be a non-compact $n$-manifold endowed with a codimension one transversely orientable foliation $\cF$ with compact leaves. Then $M^n$ is a fibre bundle over the real line $\IR$ with fibres given by the leaves of $\cF.$   
\end{corollary}

When these conditions are met, time has a more canonical interpretation as a non-compact parameter, whereas we encounter compact spatial leaves.

\medskip

\subsection{Carrollian Structures}~\\[5pt]
At the opposite extreme to Galilean structures, Carrollian spacetimes can be thought of as the structures in which the speed of light $c\to0$, i.e. the ultra-local limit. We may regard Carrollian structures as highlighting the geometry of classical observers.

\begin{definition}
A \emph{Carrollian structure} on an $n$-dimensional manifold $M^n$ is a pair $(Z, h)$ of a nowhere-vanishing vector field $Z \in \mathsf{\Gamma}(TM^n),$ called the \emph{vector field of observers}, and a corank one positive-semidefinite tensor $h \in \mathsf{\Gamma}(\midodot^2\, T^*M^n)$ such that $\iota_Z h =0,$
called the \emph{spatial metric}. A \emph{Carrollian manifold} (or \emph{Carrollian spacetime}) is a manifold endowed with a Carrollian structure.
\end{definition}

Since $\ker(h) = \mathrm{Span}(Z) \subset TM^n,$ the spatial metric $h$ is a metric on the dual of the annihilator\footnote{The \emph{annihilator} ${\rm Ann}(Z)$ of a vector field $Z$ is the subbundle of $T^*M^n$ whose sections consist of one-forms $\alpha$ that are annihilated by $Z$, i.e. $\iota_Z \alpha = 0$.} ${\rm Ann}(Z)^*\subset TM^n$ of $Z$.

\begin{remark} \label{rmk:GC}
Let $V$ and $V^*$ be the vector spaces of Remark \ref{rmk:GG} with the same choices of bases. The structure group of a Carrollian manifold is the subgroup $\sfG_\ttC \subset \mathsf{GL}(n,\IR)$ preserving $H \in V$ and $\delta_{ab} \, \sigma^a \odot \sigma^b$:
\be 
\sfG_\ttC = \Big\{  
{\small \bigg(\begin{matrix}
1 & v^\mathtt{t} \\
0 & \sfA
\end{matrix}\bigg) }\normalsize \ \Big|
 \ v \in \IR^{n-1} \ , \ \sfA \in \sfO(n-1)  \Big\} \ ,
\ee
with Lie algebra
\be  
\frg_\ttC = \Big\{  
{\small \bigg( \begin{matrix}
0 & v^\mathtt{t} \\
0 & A
\end{matrix} \bigg) } \normalsize \ 
\Big| \ v \in \IR^{n-1} \ , \ A \in \mathfrak{so}(n-1)  \Big\} \ .
\ee
The group $\sfG_\ttC$ is also isomorphic to the semi-direct product $\sfO(n-1) \ltimes \IR^{n-1}$.
\end{remark}

\begin{remark}
The Carrollian structure group $\sfG_\ttC$ has two connected components, corresponding to the value of the determinant of $\sfA \in \sfO(n-1).$ Let $\sfG_{\ttC \, 0}$ be the component connected to the identity, which is isomorphic to $\sfS\sfO(n-1) \ltimes \IR^{n-1}$. If the $\sfG_\ttC$-structure defining a Carrollian structure can be reduced to a $\sfG_{\ttC \, 0}$-structure, then there is one more characteristic tensor given by a volume form $\mu \in \mathsf{\mathsf{\Omega}}^n(M^n).$ This corresponds to the $\sfG_{\ttC \, 0}$-invariant tensor $\eta \wedge \sigma^1 \wedge \cdots \wedge \sigma^{n-1} \in \midwedge^n\, V^*.$ If the $\sfG_\ttC$-structure does not reduce further, then a volume form that can be expressed in this way only exists locally.  
\end{remark}

Similarly to the Galilean case, we say a linear connection $\nabla^\ttC$ is a compatible Carrollian connection if
\be  
\nabla^\ttC Z = 0 \qquad \mathrm{and} \qquad  \nabla^\ttC h = 0 \ .
\ee
Following \cite{Figueroa2020}, we may define a vector bundle morphism 
\be  
\Phi \colon \midwedge^2\, T^*M^n \otimes TM^n \longrightarrow \midodot^2\, {\mathrm{Ann}}(Z)
\ee
covering the identity by
\be  
\Phi(T)(X,Y) := h \bigl(T(Z, X), Y \bigr) + h \bigl(T(Z, Y), X \bigr) \ ,
\ee
for all $T \in \mathsf{\Gamma}(\midwedge^2\, T^*M^n \otimes TM^n)$ and $X, \, Y \in \mathsf{\Gamma}(TM^n).$
Then an easy calculation shows~\cite{Figueroa2020}
 
\begin{proposition} \label{prop:Zh}
Let $(Z, h)$ be a Carrollian structure on $M^n$ together with a Carrollian connection $\nabla^\ttC$, and denote its torsion by $T^{\nabla^\ttC} \in \mathsf{\Gamma}(\midwedge^2\, T^*M^n \otimes TM^n).$ Then
\be  
\Phi (T^{\nabla^\ttC}) = \pounds_Z h \ ,
\ee 
where $\pounds_Z$ denotes the Lie derivative along $Z$.
\end{proposition}

The properties of the volume form $\mu \in \mathsf{\Omega}^n(M^n)$ may be analysed by introducing the tensor $S \in \mathsf{\Gamma}(T^*M^n \otimes TM^n)$ defined by
\be \label{def:S}
S(X) \coloneqq T^{\nabla^\ttC}(Z, X) \ ,
\ee
for all $X \in \mathsf{\Gamma}(TM^n).$ Then one may prove~\cite{Figueroa2020}

\begin{proposition}\label{prop:Zmu}
Let $(Z, h)$ be a Carrollian structure on $M^n$ and let $\mu$ be its local volume form. Then
\be  
\pounds_Z \mu = {\mathrm{tr}}(S) \, \mu \ .
\ee 
\end{proposition}

\medskip

\subsection{Aristotelian Structures} ~\\[5pt]
By merging together the kinematics of Newtonian physics with the geometry of classical observers, we obtain  Aristotelian spacetimes. These are sometimes called `absolute spacetimes' and are characteristed by the absence of any boost symmetry.

\begin{definition}
An \emph{Aristotelian structure} on an $n$-dimensional manifold $M^n$ is a quadruple $(\tau, Z, \gamma, h)$ where $\tau \in \mathsf{\Omega}^1(M^n)$ is the clock form, $Z \in \mathsf{\Gamma}(TM^n)$ is the vector field of observers with $\iota_Z \tau = 1,$ $\gamma \in \mathsf{\Gamma}(\midodot^2\, TM^n)$ is the spatial cometric with $\iota_\tau \gamma = 0 ,$ and $h \in \mathsf{\Gamma}(\midodot^2\, T^*M^n)$ is the spatial metric with $\iota_Z h = 0.$ An \emph{Aristotelian manifold} (or \emph{Aristotelian spacetime}) is a manifold endowed with an Aristotelian structure.
\end{definition}

Note that $h$ defines a metric on the distribution $\ker(\tau).$ Similarly $\gamma$ defines a metric on the annihilator $\mathrm{Ann}(Z)$ of $Z$. There are splittings
\be  
TM^n = \ker(\tau)\oplus {\mathrm{Span}}(Z)
\qquad \text{and} \qquad
T^*M^n = {\mathrm{Ann}}(Z) \oplus {\mathrm{Span}}(\tau) \ .
\ee

\begin{remark}
The structure group $\sfG_\ttA\subset\mathsf{GL}(n,\IR)$ of an Aristotelian spacetime is given by the intersection of $\sfG_\ttG$ and $\sfG_\ttC,$  since an Aristotelian structure is simultaneously a Galilean structure and a Carrollian structure (see Remarks \ref{rmk:GG} and \ref{rmk:GC}). Hence
\be  
\sfG_\ttA = \sfG_\ttG \cap \sfG_\ttC = \Big\{  
{\small \bigg(\begin{matrix}
1 & 0 \\
0 & \sfA
\end{matrix} \bigg)} \normalsize \ 
\Big| \ \sfA \in \sfO(n-1) \Big\}
\ee
with Lie algebra
\be  
\frg_\ttA = \Big\{  
{\small \bigg(\begin{matrix}
0 & 0 \\
0 & A
\end{matrix} \bigg)} \normalsize \ 
\Big| \ A \in \mathfrak{so}(n-1)  \Big\} \ .
\ee
Thus $\sfG_\ttA \simeq \sfO(n-1)$ and $\frg_\ttA \simeq \mathfrak{so}(n-1).$ 
\end{remark}

Assume that there exists a linear connection preserving the Aristotelian structure, i.e. a compatible Aristotelian connection $\nabla^\ttA$. Then the properties of the tensors characterising the Aristotelian structure are given by the corresponding properties for Galilean structures (Proposition \ref{prop:Ztau}) together with the properties for Carrollian structures (Propositions \ref{prop:Zh} and \ref{prop:Zmu}).
It can be furthermore shown that~\cite{Figueroa2020}

\begin{proposition} \label{prop:tauari}
Let $M^n$ be a manifold endowed with an Aristotelian structure $(\tau, Z, \gamma, h).$ Then
\be  
\pounds_Z \tau = \tau \circ S \ ,
\ee
where the tensor $S$ is defined in Equation~\eqref{def:S}.
\end{proposition}

\begin{remark}
For an Aristotelian structure, the vector field of observers $Z$ is an infinitesimal symmetry of the clock form $\tau$, i.e. $\pounds_Z \tau = 0$, if  $\de \tau = 0.$ This means that any observer is always synchronised with the local absolute time. If only the more general integrability condition $\tau \wedge \de \tau = 0$ is imposed, then $\pounds_Z \tau \neq 0$ in general.    

For a foliated Aristotelian spacetime $M^n$, Theorem \ref{thm:Globalreeb} and Corollary \ref{cor:Reeb} imply that the vector field of observers $Z$ can be obtained as a horizontal lift of a vector field on the base manifold, since in these cases $M^n$ is a fibred manifold over either $S^1$ or $\IR$ respectively. Hence there may be different interpretations of absolute time depending on the horizontal lift and its holonomy.
\end{remark}

\begin{remark} \label{rmk:transvfol}
The foliation $\cF$ of any integrable Aristotelian spacetime $M^n$ is always transversely parallelisable, since the vector field of observers $Z$ is nowhere-vanishing. This implies that all the leaves of $\cF$ have trivial holonomy \cite{Moerdijk2003}.  

Hence $(M^n, \cF)$ can be given the structure of a Riemannian foliation \cite{Moerdijk2003}. The corresponding transverse Riemannian metric $g_{\textrm{\tiny$\perp$}}$ must satisfy
\be  
\pounds_X g_{\textrm{\tiny$\perp$}} = 0 \ , 
\ee
for all $X \in \mathsf{\Gamma}(T\cF)$. Thus it cannot simply be constructed by using $\tau$ alone, i.e. $g_{\textrm{\tiny$\perp$}} \neq \tau \otimes \tau,$ since  
\be  
\pounds_X (\tau\otimes\tau) = - 2\, (\iota_X \alpha) \, \tau \otimes \tau
\ee
by Equation~\eqref{eq:alphaint}, for all $X \in \mathsf{\Gamma}(T\cF).$ Aristotelian geometries admitting a Riemannian foliation appear in the curved spacetime fracton theories of~\cite{Bidussi:2021nmp} (for the case $\de\tau=0$).
\end{remark}

\section{The Godbillon-Vey Class of a Non-Lorentzian Spacetime}\label{sect:gvclass}

In this section we introduce the Godbillon-Vey class of a non-Lorentzian spacetime. For this, we assume that our spacetime structure always admits a clock form $\tau$ satisfying the Frobenius integrability condition $\tau \wedge \de \tau = 0.$ In other words, our spacetime manifold always admits a foliation determining the spatial leaves. This means that a classification based on the Godbillon-Vey class is possible only for integrable Galilean and Aristotelian structures.

To discuss the meaning of the Godbillon-Vey class for a non-Lorentzian manifold, we will determine its relationship with the torsion of a connection preserving the spacetime structure. This is made possible by considering how a $\sfG$-structure with an Ehresmann connection
induces a spacetime structure with adapted connection \cite{Spivak1970}. 

\medskip

\subsection{The Godbillon-Vey Class of a Foliated Manifold} ~\\[5pt]
As our interest in the following is in features of foliated non-Lorentzian spacetimes, we will mainly focus on the description of the Godbillon-Vey class of foliated manifolds of codimension one \cite{Godbillon1971}, following \cite{Moerdijk2003}.

Let $M^n$ be an $n$-dimensional manifold. Recall that a codimension one foliation of $M^n$ is defined by a nowhere-vanishing one-form $\tau \in \mathsf{\Omega}^1(M^n)$ which is {integrable}, in the sense that Equation~\eqref{eq:integrability} holds:
\be  
\tau\wedge\de\tau = 0 \ .
\ee
This implies Equation~\eqref{eq:alphaint}:
\be   
\de \tau =  \tau \wedge \alpha \ ,
\ee
for some (not unique) $\alpha \in \mathsf{\Omega}^1(M^n).$ 

\begin{lemma} \label{LemmaTau}
The one-form $\alpha$ satisfying Equation~\eqref{eq:alphaint} obeys
\be  
\de \alpha \wedge \tau = 0
\qquad \mbox{and} \qquad
\de (\alpha \wedge \de \alpha) = 0 \ .
\ee
\end{lemma}

\begin{proof}
It follows from Equation~\eqref{eq:integrability} that
\begin{align*}
0 = \de (\alpha \wedge \tau) =  \de \alpha \wedge \tau - \alpha \wedge \de \tau   = \de \alpha \wedge \tau - \alpha \wedge \alpha \wedge \tau =   \de \alpha \wedge \tau  \ .
\end{align*}
Since $\de \alpha \wedge \tau = 0,$ it follows that
\be  
\de \alpha = \beta \wedge \tau \ ,
\ee
for some $\beta \in \mathsf{\Omega}^1(M^n).$ Hence
\be  
\de (\alpha \wedge \de \alpha) = \de \alpha \wedge \de \alpha = \beta \wedge \tau \wedge \beta \wedge \tau = 0 \ ,
\ee
and the result follows.
\end{proof}

Lemma~\ref{LemmaTau} suggests 

\begin{definition}
Let $M^n$ be an $n$-manifold with a codimension one foliation defined by a nowhere-vanishing one-form $\tau.$ Its \emph{Godbillon-Vey class} $\gv(\tau)$ is the de~Rham class
\be  
\gv(\tau) \coloneqq [\alpha \wedge \de \alpha] \ \in \ \sfH^3(M^n;\IR) \ .
\ee
\end{definition}

\begin{remark}
It is easy to show that the Godbillon-Vey class $\gv(\tau)$ is independent of the choice of $\alpha.$ Let $\alpha^\prime \in \mathsf{\Omega}^1(M^n)$ be another one-form satisfying 
\be  
\de \tau = \tau \wedge \alpha^\prime \ .
\ee
Then it follows that
\be  
\tau \wedge (\alpha^\prime - \alpha)= 0 \ ,
\ee
which implies
\begin{align}  
\alpha^\prime - \alpha = f\, \tau \ , 
\end{align}
for some $f \in C^\infty(M^n).$  From Lemma \ref{LemmaTau} it follows that 
\be  
\tau \wedge \de \alpha^\prime =0 \ , 
\ee
and hence 
\begin{align}
\alpha^\prime \wedge \de \alpha^\prime  = (\alpha + f\, \tau) \wedge \de (\alpha + f\, \tau)   = \alpha \wedge \de \alpha + \alpha \wedge \de (f\, \tau) \ .  
\end{align} 
Since
\be  
\de (\alpha \wedge f\, \tau) = \de \alpha \wedge f\, \tau - \alpha \wedge \de (f\, \tau) \ ,
\ee
it follows by Lemma \ref{LemmaTau} that
\be  
\de (\alpha \wedge f\, \tau) = -\alpha \wedge \de (f\, \tau) \ .
\ee
Therefore 
\be  
\alpha^\prime \wedge \de \alpha^\prime = \alpha \wedge \de \alpha - \de (\alpha \wedge f\, \tau)  \ ,
\ee
proving that the Godbillon-Vey class does not depend on the choice of $\alpha.$
\end{remark}

\begin{lemma} \label{LemmaTau2}
Let $(M^n, \tau)$ be an $n$-dimensional manifold with a codimension one foliation and let $f \in C^\infty(M^n)$ be a nowhere-vanishing function. Then $\gv(\tau) = \gv (f \,\tau).$
\end{lemma}

\begin{proof}
The calculation
\begin{align*}
\de (f\, \tau) = \de f \wedge \tau + f\, \de \tau  = \tfrac1f \, \de f \wedge f\, \tau + f\, \tau \wedge \alpha  = f\, \tau \wedge ( \alpha -  \de\log |  f | )  
\end{align*}
shows that $f\, \tau$ is integrable as well. This yields
\be  
\gv(f\, \tau) = \big[ ( \alpha - \de \log | f | ) \wedge  \de ( \alpha - \de \log | f | ) \big] \ .
\ee
It is then straightforward to see that
\be  
( \alpha - \de \log | f | ) \wedge  \de ( \alpha - \de \log | f |)= \alpha \wedge \de \alpha - \de (\log | f | \, \de \alpha) \ .
\ee
Thus $\gv(f\, \tau) =  \gv(\tau).$
\end{proof}

\begin{remark}
The integrable one-forms $\tau$ and $f\, \tau$ define the same foliation $\cF,$ because $\ker(\tau) = \ker(f\, \tau) =  T\cF.$ 
Hence Lemma \ref{LemmaTau2} shows that the Godbillon-Vey class does not depend on the choice of $\tau,$ and it is indeed an invariant of the foliation $\cF$ itself. We can therefore call $\gv(\cF) \in \sfH^3(M^n;\IR)$ the Godbillon-Vey class of the foliation $\cF$ without any reference to the non-unique integrable one-form $\tau.$
\end{remark}

\begin{remark}
The restriction of the one-form $\alpha$ to the leaves of $\cF$ are
closed forms which define a leafwise cohomology class  $[\alpha]\in\sfH^1(M^n,\cF)$,
called the \emph{Reeb class}. It is an obstruction to the existence of
a globally defined transverse volume form.
\end{remark}

\begin{remark}\label{rmk:alphaLie}
Assume that there exists a vector field $Z \in \mathsf{\Gamma}(TM^n)$ such that $\iota_Z \tau =1.$ Then the integrability condition gives
\be  
0 = \iota_Z (\tau\wedge\de\tau) = \de \tau - \tau \wedge (\iota_Z\, \de \tau)  \ .
\ee
From
\be  
\pounds_Z \tau = \iota_Z\, \de \tau 
\ee
it follows that 
\be  
\de \tau = \tau \wedge (\pounds_Z \tau) \ ,
\ee
so that $Z$ determines a choice of one-form $\alpha$ given by
\be \label{eq:alphachoice}
\alpha = \pounds_Z \tau + f_Z\, \tau \ ,
\ee
where $f_Z \coloneqq \iota_Z \alpha \in C^\infty(M^n).$ 

By taking the differential of Equation \eqref{eq:alphachoice}, we obtain
\be
\de \alpha =  \tau \wedge \pounds_Z^2 \tau + f_Z\, \de \tau + \de f_Z \wedge \tau \ .
\ee
We then find
\be
\alpha \wedge \de \alpha = - \tau \wedge \pounds_Z \tau \wedge \pounds_Z^2\tau + \de (f_Z\, \de \tau)\ .
\ee
It follows that the Godbillon-Vey class $[\alpha\wedge\de\alpha]$ does not depend on the function $f_Z,$  hence we may as well set $f_Z=0$. 

Then
\be \label{eq:Zalpha}
\iota_Z \alpha = 0 \ ,
\ee
so $\alpha$ cannot be a section of ${\mathrm{Span}}(\tau) \subset T^*M^n.$ Hence it belongs to the dual of $\mathsf{\Gamma}(T\cF).$
In this case, the calculation above shows that the Godbillon-Vey class is represented by the three-form
\be  
\alpha \wedge \de \alpha = - \tau \wedge \pounds_Z \tau \wedge \pounds_Z^2\tau \ . 
\ee
\end{remark}

\begin{example}[\textbf{Thurston's Helical Wobble}]\label{rmk:helical}
Thurston gave the geometric interpretation of the Godbillon-Vey class~\cite{Thurston1972}, represented by the phenomenon he called \emph{helical wobble}. We will focus on three-dimensional circle bundles over hyperbolic surfaces admitting a codimension one foliation, following \cite{Calegari2007}.

Let $M$ be a manifold with fundamental group $\pi_1(M)$ based at a certain point $x \in M$, and let \smash{$\widetilde{M}$} be the universal cover of $M.$ Suppose there exists a manifold $F$ with a left $\pi_1(M)$-action. Then we can construct the bundle
\be  
E = \widetilde{M} \times_{\pi_1(M)} F
\ee
over $M$ whose fibres are identified with $F.$ It is obtained from $\widetilde{M} \times F$ with the identifications
$(R_\gamma(p), e) \sim (p, L_\gamma(e))$ for all $\gamma\in\pi_1(M)$, \smash{$p \in \widetilde{M}$} and $e \in F,$
where $R_\gamma$ is the right $\pi_1(M)$-action on $\widetilde{M}$ by covering automorphisms and $L_\gamma$ denotes the left $\pi_1(M)$-action on $F.$

The foliation given by the fibres of the projection
$\mathrm{pr}_2 \colon \widetilde{M} \times F \rightarrow F$ is invariant under the action of $\pi_1(M).$ Hence it induces a foliation $\cF$ of $E.$ The connected components of this foliation are diffeomorphic to $M.$ 

Consider now the special case where $M = \Sigma$ is a surface and $F=S^1$ with the action of the fundamental group $\pi_1(\Sigma)$ on $S^1$ given by the representation 
$\rho \colon \pi_1(\Sigma) \rightarrow {\mathsf{Homeo}}(S^1)$ as homeomorphisms of the circle. Then the quotient \smash{$E = \widetilde{\Sigma}  \times_{\pi_1(\Sigma)} S^1 $} determines a surjective submersion $\pi \colon E \rightarrow \Sigma$ whose fibres are copies of the circle $S^1.$
The action of $\pi_1(\Sigma)$ on \smash{$\widetilde{\Sigma}$} preserves the foliation $\cF$ whose leaves are $\Sigma \times \{ z\},$ for $z \in S^1.$ Hence $\cF$ descends to $E$ as a codimension one foliation transverse to the circle fibres, i.e.~$E$ is a foliated circle bundle.

We further assume that $\Sigma$ is a complete hyperbolic surface. Then $E$ can be endowed with a harmonic measure given by a transverse volume form, such that the measure of a transversal vector field is preserved on average by holonomy transport along a path on a leaf of $\cF$ which covers a random walk on $\Sigma.$ By integrating the harmonic measure on the circle fibres, we obtain a metric on the fibres, which together with the pullback of the metric $h$ on $\Sigma$ defines a metric $g$ on $E.$
In this case $\alpha \in \mathsf{\Gamma}(T^*\cF)$ measures the logarithmic derivative of the transverse volume form under holonomy \cite{Calegari2007}.

By using $g$ we obtain a vector field $\alpha^\sharp \coloneqq \iota_\alpha g^{-1}  \in \mathsf{\Gamma}(T\cF).$ Then the Godbillon-Vey class $\gv(\cF) = [\alpha \wedge \de \alpha] \in \sfH^3(E;\IR)$ measures the infinitesimal rate at which $\alpha^\sharp$ spins while moving transversely to $\cF.$ This is the phenomenon of {helical wobble} of the foliation $\cF$, which is analogous to the `wobble' of spinning rigid bodies due to the tilt between their axis of symmetry and their angular momentum.
\end{example}

\medskip

\subsection{Intrinsic Torsion of Spacetime Structures}~\\[5pt] \label{subs:Intrinsic}
We will describe the Godbillon-Vey class of a non-Lorentzian spacetime by using its \emph{intrinsic torsion}. Instrinsic torsion is the part of the torsion of a compatible connection that depends only on the underlying $\sfG$-structure; it is the first order obstruction to integrability of a $\sfG$-structure, i.e. to the existence of an open cover of $M^n$ such that the restriction of the $\sfG$-structure to each open set is isomorphic to the standard flat $\sfG$-structure on a model vector space $V = \IR^n$. We shall start by recalling some basic properties of $\sfG$-structures arising from the reduction of the frame bundle $\sfF(M^n)$ of a spacetime $M^n$.

Let $M^n$ be an $n$-dimensional manifold. Its frame bundle $\pi \colon \sfF(M^n) \rightarrow M^n$ is a principal $\mathsf{GL}(n, \IR)$-bundle whose points are given by the choice of a frame at $x \in M^n,$ i.e. a basis for the tangent space $T_xM^n.$ A frame is thus interpreted as a map $u \colon \IR^n \rightarrow T_xM^n$ with projection~$\pi(u) = x .$
The frame bundle $\sfF(M^n)$ is naturally endowed with a  \emph{soldering form} $\vartheta \in \mathsf{\Omega}^1(\sfF(M^n), \IR^n)$ defined by
\be  
\iota_{X_u} \vartheta_u \coloneqq u^{-1}(\pi_* (X_u)) \ , 
\ee
for all $X_u \in T_u \sfF(M^n).$

If we reduce $\sfF(M^n)$ to a principal $\sfG$-bundle $\pi \colon P \rightarrow M^n,$\footnote{Here we slightly abuse notation by denoting the projection with $\pi$ again.} where $\sfG$ is either of the subgroups $\sfG_\ttG$ or $\sfG_\ttA$ of $\mathsf{GL}(n, \IR),$ then $P$ inherits the soldering form $\vartheta$ from $\sfF(M^n).$
Since $\vartheta$ is basic, it determines an isomorphism $P \times_\sfG V \simeq TM^n,$ where $V= \IR^n,$ which extends to $P \times_\sfG V^* \simeq T^*M^n$ as well as all tensor products.  

Let $\frg$ be the Lie algebra of $\sfG$. A choice of Ehresmann connection $\omega \in \mathsf{\Omega}^1(P, \frg)$ yields an associated linear connection $\nabla$ on $TM^n$ preserving the tensors defined by the $\sfG$-structure. The \emph{intrinsic torsion} $\Theta \in \mathsf{\Omega}^2_\sfG(P, \IR^n)$ of the connection 
$\omega$ is given by
\be  
\Theta =  \de \vartheta + \omega \wedge \vartheta \ ,
\ee
where in $\omega \wedge \vartheta$ the Lie algebra $\frg$ acts on $\IR^n$ via the embedding $\frg \subset \mathfrak{gl}(n, \IR).$ Then the torsion $T^\nabla$ of the associated linear connection $\nabla$ is given by
\be \label{eq:TorsionTheta}
\pi^* ( T^\nabla_x ) = u \circ \Theta_u \ ,
\ee
at the point $x \in M^n$ with $\pi(u) = x.$ In other words, if $X_x, Y_x \in T_x M^n$ are vectors whose horizontal lift is given by $\bar{X}_u, \bar{Y}_u\in T_uP,$ then 
\be  
T^\nabla(X_x , Y_x) = u(\iota_{\bar{Y}_u}\, \iota_{\bar{X}_u} \Theta) \ .
\ee
In this formulation, intrinsic torsion is the first order obstruction to the existence of an atlas of coordinate charts of $M^n$ whose canonical frame fields are $\sfG$-frames. 

As discussed in \cite{Figueroa2020}, the intrinsic torsion of a spacetime structure can be characterised by the Spencer differential
\be  
\del \colon {\mathsf{Hom}}(V , \frg) \longrightarrow {\mathsf{Hom}}(\midwedge^2\, V, V) \ ,
\ee
where $V = \IR^n,$ defined by
\be  
\del = (\unit_V \otimes \wedge) \circ (i \otimes \unit_{V^*}) 
\ee
under the identifications ${\mathsf{Hom}}(V , \frg) \simeq \frg \otimes V^*$ and $ {\mathsf{Hom}}(\midwedge^2\, V, V) \simeq \midwedge^2\, V^* \otimes V,$ where the map $i \colon \frg \rightarrow  V\otimes V^*$  is the embedding $\frg \subset \mathfrak{gl}(V)$ composed with the isomorphism $\mathfrak{gl}(V) \simeq V \otimes V^*.$
The Spencer differential $\del$ yields the exact sequence of $\sfG$-equivariant maps
\be  
0 \longrightarrow \ker(\del) \longrightarrow \frg \otimes V^* \xlongrightarrow{\del} \midwedge^2\, V^* \otimes V \longrightarrow {\mathrm{coker}}(\del) \longrightarrow 0 
\ee 
where ${\mathrm{coker}}(\del) \coloneqq (\midwedge^2\, V^* \otimes V) / {\mathrm{im}}(\del).$

From Equation~\eqref{eq:TorsionTheta} it follows that $\pi^*(T^\nabla) \in \mathsf{\Gamma}\bigl(\midwedge^2\, T^*P \otimes \pi^*(TM^n) \bigr).$ Under the isomorphism $\mathsf{\Omega}^\bullet_\sfG(P) \simeq \mathsf{\Omega}^\bullet(M^n)$, and the isomorphism $\midwedge^2\, T^*M^n \otimes TM^n \simeq P \times_\sfG (\midwedge^2\, V^* \otimes V)$ induced by the soldering form $\vartheta,$ it may be shown that \cite{Figueroa2020}
$\pi^* ( T^\nabla )\in \mathsf{\Gamma}(P \times_\sfG {\mathrm{coker}}(\del))$. This is a consequence of $T^\nabla - T^{\nabla^\prime} = \del (\nabla - \nabla^\prime),$ where here $\del$ is extended to a vector bundle morphism 
\be  
\del \colon P \times_\sfG (\frg \otimes V^*) \longrightarrow P \times_\sfG (\midwedge^2\, V^* \otimes V) \ ,
\ee
which is possible because it is $\sfG$-equivariant.  

Spacetime structures are then classified according to the number of $\frg$-submodules of the $\frg$-module ${\mathrm{coker}}(\del).$
Following \cite{Figueroa2020} we will discuss two instances of the classification of Aristotelian structures $(\tau, Z, \gamma, h)$ that are particularly relevant to this paper:

\begin{myitemize}
\item The existence of a spatial foliation $\cF$ with $\de \tau \neq 0$, i.e. the realisation of the condition $\tau \wedge \de \tau = 0,$  
is equivalent to requiring that the pointwise image of \smash{$\pi^* (T^{\nabla^\ttA})$ in ${\mathrm{coker}}(\del)$} at least admits a subspace isomorphic to ${\mathrm{Span}}\big(H \otimes (\sigma^a \wedge \eta)\big)$ (cf. Remark~\ref{rmk:GG}). In other words
\be \label{eq:Tcomponents}
T^{\nabla^\ttA} \ \in \ \mathsf{\Gamma}\bigl( TM^n \otimes \big(\mathrm{Ann}(Z) \wedge \mathrm{Span}(\tau)\big) \bigr)  \ ,
\ee
where the $TM^n$-component will always admit a $\mathrm{Span}(Z)$-component and
$\mathrm{Ann}(Z) \simeq T^*\cF.$
\item Consider an integrable Aristotelian structure characterised by $ \de \tau \neq 0$ and $\pounds_Z \mu=0.$ Then
\be \label{eq:classmuT}
T^{\nabla^\ttA} \ \in \ \mathsf{\Gamma}\big(  \mathrm{Span}(Z)\otimes \big(\mathrm{Ann}(Z)\wedge \mathrm{Span}(\tau)\big) \, \oplus \, \ker(\tau) \odot_0 \bigl( \mathrm{Ann}(Z) \wedge \mathrm{Span}(\tau) \bigr) \big) \ ,
\ee
where $\odot_0$ denotes the traceless symmetric tensor product.
\end{myitemize}

\medskip
 
\subsection{Godbillon-Vey Invariants of Spacetime Structures}~\\[5pt]
To characterise the Godbillon-Vey class of a spacetime structure, we shall determine an expression for the one-form $\alpha$ from Equation~\eqref{eq:alphaint} by relating it to the intrinsic torsion of our $\sfG$-structure. 

Let $(M^n,\tau, \gamma)$ be an integrable Galilean manifold and let $\cF$ be the codimension one foliation of $M^n$ determined by the clock form $\tau$. Let the principal $\sfG_\ttG$-bundle $\pi\colon P \rightarrow M^n$ be the $\sfG_\ttG$-structure determining the Galilean structure. Then the foliation $\cF$ is transversal to the surjective submersion $\pi,$ since 
\be  
(\pi_*)_u (T_u P) + T_{\pi(u)} \cF = T_{\pi(u)}M^n \ ,
\ee
for any $u \in P.$

Hence $P$ is endowed with the pullback foliation $\pi^*(\cF)$ satisfying
\be \label{eq:transverse}
T(\pi^*(\cF)) = \pi_*^{-1}(T\cF) 
\ee
and
\be  
\mathrm{codim}(\pi^*(\cF))=\mathrm{codim}(\cF)= 1 \ .
\ee
It is easy to show that $\pi^* \tau \in \mathsf{\Omega}^1(P)$ determines $\pi^*(\cF)$: It follows straightforwardly from Equation~\eqref{eq:transverse} that
\be  
\ker(\pi^*\tau)= T(\pi^*(\cF)) \ .
\ee
Therefore the Frobenius integrability condition for $\pi^*(\cF)$ reads as
\be  
\de(\pi^* \tau) = \pi^* \de  \tau = \pi^* \alpha \wedge \pi^* \tau \ ,
\ee
which determines its Godbillon-Vey class in $\sfH^3(P;\IR)$ through pullback
\be  
\gv(\pi^*(\cF))=[\pi^*\alpha \wedge \de \pi^* \alpha] = [\pi^*(\alpha \wedge \de \alpha)] = \pi^*\gv(\cF) \ .
\ee

\begin{proposition}
Let $(\tau, \gamma)$ be an integrable Galilean structure on an $n$-dimensional manifold $M^n$ determined by a $\sfG_\ttG$-structure $\pi \colon P \rightarrow M^n$ with Ehresmann connection $\omega,$ and endowed with a compatible Galilean connection $\nabla^\ttG$.  Then the one-form $\alpha$ satisfies
\be \label{eq:psitauTheta}
\pi^*(\tau \wedge \alpha)_u = \langle u (\Theta_u) , \tau_x \rangle_{T_xM^n} \ ,
\ee
where $u \in P$ with $\pi(u) = x \in M^n,$ \smash{$\Theta \in \mathsf{\Omega}^2_{\sfG_\ttG}(P , \IR^n)$} is the intrinsic torsion of the connection $\omega$ and the right-hand side denotes the natural duality pairing between the $T_xM^n$-component of $u(\Theta_u) \in \midwedge^2\, T_u^*P \otimes T_xM^n$ and $\tau_x \in T_x^*M^n.$
\end{proposition}

\begin{proof}
By combining Equations~\eqref{eq:alphaint}, \eqref{eq:detau} and \eqref{eq:TorsionTheta} we obtain
\be  
\langle u( \iota_{\bar{Y}_u}\, \iota_{\bar{X}_u} \Theta_u) , \tau_x \rangle_{T_xM^n} = \iota_{Y_x}\, \iota_{X_x} (\tau \wedge \alpha)_x  \ ,
\ee
where $\bar{Y}_u, \, \bar{X}_u \in T_u P$ are the horizontal lifts of $Y_x , \, X_x \in T_xM^n$ respectively. By using
\be  
 \iota_{Y_x}\, \iota_{X_x} (\tau \wedge \alpha)_x =  \iota_{\bar{Y}_u} \, \iota_{\bar{X}_u} \pi^*(\tau \wedge \alpha)_u \ ,
\ee
the expression \eqref{eq:psitauTheta} then follows.
\end{proof}

\begin{remark} \label{rmk:thetabar}
Following \cite{Figueroa2020}, define the two-form $\bar{\Theta} \in \mathsf{\Omega}^2_{\sfG_\ttG}(P)$ by 
\be  
\bar{\Theta}_u \coloneqq \langle u (\Theta_u) , \tau_x \rangle_{T_xM^n}
\ee
for all $u \in P$ with $\pi(u)=x.$ Then Equation~\eqref{eq:psitauTheta} reads
\be  
\bar{\Theta} = \pi^* (\tau \wedge \alpha) \ ,
\ee
which is the pullback of the two-form
\be \label{eq:Talph}
\bar{T}^{\nabla^\ttG} := \tau\circ T^{\nabla^\ttG} = \de\tau = \tau \wedge \alpha 
\ee
on the Galilean manifold $M^n$, obtained by combining Equations~\eqref{eq:alphaint} and \eqref{eq:detau}.
\end{remark}

In the following we will show how the Godbillon-Vey class of the pullback foliation of an integrable Aristotelian spacetime is related to the intrinsic torsion of its defining $\sfG_\ttA$-structure. In order to obtain an expression for $\alpha$ depending only on the tensors characterising the spacetime structure, i.e. to solve Equation~\eqref{eq:psitauTheta}, we need more data. We show that Equation~\eqref{eq:psitauTheta} can be solved for Aristotelian structures.

\begin{lemma} \label{lemma:pial}
Let $(\tau, Z , \gamma, h)$ be an integrable Aristotelian structure on $M^n$ with $\sfG_\ttA$-structure $\pi \colon P \rightarrow M^n.$ Then 
\be \label{eq:alphaTheta}
\pi^* \alpha_u = \langle u (\iota_{\bar{Z}_u} \Theta_u) , \tau_x \rangle_{T_xM^n} \ ,
\ee
where $\bar{Z} \in \mathsf{\Gamma}_{\sfG_\ttA}(TP)$ is the horizontal lift of $Z \in \mathsf{\Gamma}(TM^n).$
\end{lemma}

\begin{proof}
By contracting both sides of Equation~\eqref{eq:psitauTheta} with $\bar{Z}_u \in T_uP,$ the horizontal lift of $Z_x \in T_xM^n,$ on the left-hand side we find
\be  
\iota_{\bar{Z}} \pi^* (\tau \wedge \alpha) = \pi^* \bigl((\iota_Z \tau) \, \alpha - (\iota_Z \alpha) \, \tau \bigr) = \pi^* \alpha \ ,
\ee
where the last equality follows from Equation~\eqref{eq:Zalpha}. Then Equation~\eqref{eq:alphaTheta} follows straightforwardly.
\end{proof}

\begin{remark}
Following Remark \ref{rmk:thetabar} we can obtain the counterpart of Lemma \ref{lemma:pial} on the Aristotelian manifold $M^n.$ Solving Equation~\eqref{eq:Talph} by contracting both sides with the vector field of observers $Z \in \mathsf{\Gamma}(TM^n)$ we obtain
\be \label{eq:alpT}
\alpha= \iota_Z \bar{T}^{\nabla^\ttA} \ ,
\ee
and the result of Lemma \ref{lemma:pial} is the pullback of Equation~\eqref{eq:alpT}.\footnote{A simpler argument recalls from Remark~\ref{rmk:alphaLie} that we can choose $\alpha =\pounds_Z \tau $
and applies Proposition~\ref{prop:tauari} to get \smash{$\iota_Z \bar{T}^{\nabla^\ttA} = \pounds_Z \tau$}.}
\end{remark}

The characterisation of the Godbillon-Vey class for an integrable Aristotelian structure in terms of its intrinsic torsion is now completed as

\begin{proposition} \label{prop:gvpull}
Let $(\tau, Z , \gamma, h)$ be an integrable Aristotelian structure on $M^n$ induced by the $\sfG_\ttA$-structure $\pi \colon P \rightarrow M^n$ with Ehresmann connection $\omega \in \mathsf{\Omega}^1(P, \frg_\ttA).$ Then
\be \label{eq:GVTheta}
\pi^*(\alpha \wedge \de \alpha) = \iota_{\bar{Z}} \bar{\Theta} \wedge \de^\omega\,\iota_{\bar{Z}}\bar{\Theta} \ ,
\ee
where $\de^\omega \colon \mathsf{\Omega}^\bullet_{\sfG_\ttA} (P) \rightarrow \mathsf{\Omega}^{\bullet+1}_{\sfG_\ttA} (P)$ is the covariant derivative induced by $\omega.$
\end{proposition}

\begin{proof}
It follows straightforwardly from Lemma \ref{lemma:pial} that
\be  
\pi^*(\alpha \wedge \de \alpha) = \iota_{\bar{Z}} \bar{\Theta} \wedge \de \, \iota_{\bar{Z}} \bar{\Theta} \ .
\ee
Recalling that
\be   
\de^\omega\, \iota_{\bar{Z}} \bar{\Theta} = \de\, \iota_{\bar{Z}} \bar{\Theta} + \omega \wedge \iota_{\bar{Z}} \bar{\Theta} \ ,
\ee
we find
\be  
\iota_{\bar{Z}} \bar{\Theta} \wedge \de^\omega\, \iota_{\bar{Z}} \bar{\Theta} = \iota_{\bar{Z}} \bar{\Theta} \wedge \de\, \iota_{\bar{Z}} \bar{\Theta} \ ,
\ee
and Equation~\eqref{eq:GVTheta} follows.
\end{proof}

\begin{remark}\label{rmk:integrGAri}
 Proposition \ref{prop:gvpull} provides an interpretation of the Godbillon-Vey class in terms of the integrability of the $\sfG_\ttA$-structure: the non-triviality of the Godbillon-Vey class $\gv(\cF) = [\alpha \wedge \de \alpha]$ for the foliation of the base manifold $M^n$ obstructs the integrability of the $\sfG_\ttA$-structure, i.e. the intrinsic torsion $\Theta$ of the principal $\sfG_\ttA$-bundle $\pi \colon P\to M^n$ cannot vanish.    
\end{remark}

\begin{remark}
We can easily relate the result of Proposition \ref{prop:gvpull} to the characterisation of the Godbillon-Vey class on the spacetime $M^n$ by the torsion of the linear Aristotelian connection $\nabla^\ttA.$ Using Equation~\eqref{eq:alpT} we find
\be  
\alpha \wedge \de \alpha = \iota_Z \bar{T}^{\nabla^\ttA} \wedge \de^{\nabla^\ttA}\,\iota_Z \bar{T}^{\nabla^\ttA} \ .
\ee
In $n=3$ dimensions this can be interpreted, as we discuss further in Section~\ref{sect:IdealAri}, as a measure of the local spin of the spatial leaves in the time direction, similarly to the helical wobble from Example~\ref{rmk:helical}. In other words, the spin of the spatial leaves in the time direction is controlled by the torsion tensor of the Aristotelian structure. Noticeably, the spatial leaves of torsion-free Aristotelian spacetimes do not experience any such spin.
\end{remark}

\begin{remark}
One of the main goals of this construction is to provide a further topological ramification of the classification of non-Lorentzian spacetimes given in \cite{Figueroa2020}. In particular, the Godbillon-Vey class arising from the foliation given by the spatial leaves would yield further branches of classes for the cases when $\de\tau\neq0$ and $\tau \wedge \de \tau = 0,$ whereas the cases with $\de \tau = 0$ would not gain any new insight. We defer the completion of this task to future work.
\end{remark}

\begin{remark}
There are higher analogues of the Godbillion-Vey class defined for foliated manifolds with foliations of any codimension $q\geq 1$: Associated to a codimension $q$ foliation $\cF$ defined by a
$q$-form $\tau$ is a one-form $\alpha$ such that $\de \tau = \tau\wedge\alpha$. This has the property that $\alpha\wedge (\de\alpha)^q$ is closed, and that its de~Rham cohomology
class in degree~$2q+1$ is independent of the particular choices made for $\tau$ and $\alpha$.

These classes can be discussed in the context of more general integrable $p$-brane Newton-Cartan geometries admitting foliations of codimension~$q=p+1$, see e.g.~\cite{Pereniguez:2019eoq,Ebert:2021mfu}. The case $p=0$ corresponds to the geometries related
to particle probes which are discussed in the present paper. For $p=1$ we obtain torsional string Newton-Cartan structures~\cite{Bidussi:2021ujm}, while $p=2$ corresponds to 11-dimensional membrane Newton-Cartan geometry~\cite{Blair:2021waq}. We do not explore these extensions in this paper.
\end{remark}

\section{Topological Fluid Dynamics on Aristotelian Manifolds} \label{sect:IdealAri}

The state of an ideal fluid flowing in an oriented three-manifold $M^3$ is specified by a divergence-free vector field called the vorticity. A vorticity field which does not change with time represents an equilibrium state of the fluid flow.
Generically the helicity of the vorticity field is the only topological invariant of the fluid flow. 

In this final section we apply our previous considerations to define the notion of an `Aristotelian fluid', and study its geometric properties as well as its dynamics in detail. For these fluids the helicity is trivial. Instead, we demonstrate how the higher order invariant provided by the Godbillion-Vey class provides a novel and useful alternative to the topological characterisation of fluid flows on an integrable Aristotelian manifold, as well as their dynamics and conservation laws. This generalises and systematises previous treatments of Godbillon-Vey invariants in the fluid mechanics literature.

\medskip

\subsection{Ideal Hydrodynamics} ~\\[5pt]
We start by recalling the general formalism of ideal hydrodynamics on Riemannian manifolds, following~\cite{Arnold}, to provide a geometric picture for ideal incompressible fluids. In this setting the properties of fluid flows are encoded in a background Riemannian metric and volume form on an $n$-dimensional manifold. Although the classical settings typically take place in space dimensions $n=2$ and $n=3$, and usually on simply connected open Euclidean domains, here we shall consider the more universal setting of inviscid incompressible flows on arbitrary oriented Riemannian manifolds of any dimension. This elucidates general geometric features of the non-linear partial differential equations describing fluid flows, formulated in a unified and covariant way which allows for arbitrary background geometries. 

\begin{definition} \label{def:idfluid}
An \emph{ideal incompressible fluid} flowing in an oriented manifold $M^n$ with $\dim(M^n)=n\geq2$ is given by the data of a Riemannian metric $g$ on $M^n$, a Riemannian connection $\nabla$, and a one-parameter family of vector fields $v \in \mathsf{\Gamma}_\mu(TM^n),$ called the \emph{fluid velocity}, which preserves a volume form $\mu \in \mathsf{\Omega}^n(M^n),$ called the 
\emph{fluid density}, and evolves in a time parameter $t\in\IR$ according to the ({incompressible}) \emph{Euler equations}\footnote{If $\alpha$ is a one-form on $M^n$ then $\alpha^\sharp=\iota_\alpha g^{-1}$ denotes the dual vector field with respect to the Riemannian metric $g$. Similarly, if $v$ is a vector field on $M^n$ then $v^\flat=\iota_vg$ denotes the dual one-form. In a local coordinatisation of $M^n$, this is just the standard operation of `raising and lowering indices' using the metric~$g$.}
\be \label{eq:Euler1}
\frac{\partial v}{\partial t} + \nabla_v v = - (\de p)^\sharp
\ee 
and
\be \label{eq:Euler2}
\mathrm{div}_\mu( v ) = 0 \ ,
\ee
where $p\in C^\infty(M^n)$ is the time-dependent \emph{pressure field}. 

If $M^n$ has a non-empty boundary $\del M^n$, the velocity vector field $v$ is parallel to $\del M^n.$\footnote{\label{fn:flux} A divergence-free vector field $v$ is \emph{parallel} to $\del M^n$ if it has no {flux} through $\del M^n$, i.e. $\omega_v \rvert_{\del M^n}=0,$ where $\omega_v \coloneqq \iota_v \mu$ is the closed $n{-}1$-form associated to $v$. For the fluid velocity this means that the fluid does not flow out of the domain $M^n$.}
 
The \emph{configuration space} of an ideal incompressible fluid flowing in $M^n$ is the Lie group of volume-preserving diffeomorphisms $\mathsf{Diff}_\mu (M^n)$.\footnote{We consider only the component of $\mathsf{Diff}_\mu(M^n)$ which is connected to the identity.} If $\del M^n\neq\varnothing$, this group also preserves the boundary of $M^n$. 
\end{definition}

The fluid density $\mu$ generally differs from the Riemannian volume form induced by the metric $g$ in a positive function of proportionality; however, no specific properties of this function are assumed. The pressure $p$ is uniquely defined  (up to a time-dependent additive constant) by the Poisson equation that comes from taking the divergence of the flow equation~\eqref{eq:Euler1} and using the divergence-free constraint~\eqref{eq:Euler2} to set $\mathrm{div}_\mu\big(\frac{\partial v}{\partial t}\big)=0$. In this equation, time appears only as a parameter: the time dependence of the pressure field is dictated by the Euler equations.

The definition of the configuration space carries an implicit notion of time: time $t$ parametrizes the subgroup of diffeomorphisms $\varphi_t \in \mathsf{Diff}_\mu(M^n)$ given by the flow of the fluid velocity $v.$ A fluid particle at $x_0 \in M^n$ at the initial time $t_0$ is carried to its position $x= \varphi_t(x_0)$ at time $t$ by a one-parameter group of diffeomorphisms preserving the orientation of $M^n.$ Then the velocity of the fluid at the point $x\in M^n$ is given by $v(t,x) = \frac\del{\del t} \varphi_t (x_0).$ The integral curves of the velocity vector field $v$ are called \emph{fluid lines}, which can be regarded as geodesics on the configuration space $\mathsf{Diff}_\mu (M^n)$.

\begin{remark}
In standard three-dimensional Euclidean hydrodynamics the transport term
$\nabla_v v$ in the flow equation~\eqref{eq:Euler1} replaced by
\be  
v \times \mathrm{curl}(v) \coloneqq \star_\mu\, \big( v^\flat \wedge \mathrm{curl}(v)^\flat\big)  = -(\iota_v \, \de v^\flat)^ \sharp \ ,
\ee
whereas the divergence-free constraint \eqref{eq:Euler2} can be written as
\be \label{eq:divfree3d}
\textrm{div}_\mu(v) := \star_\mu\,\de\,\iota_v\mu = 0 \ ,
\ee
where $\star_\mu$ is the Hodge operator associated to the volume form $\mu$, i.e. the dual multi-vector field $\mu^{-1}$ regarded as a map from forms to multi-vectors.
This allows one to rewrite the flow equation \eqref{eq:Euler1} in dual form as the local conservation law
\be \label{eq:Eualt}
\frac{\De v^\flat}{\De t} = - \de \bigl( p - \tfrac{s_v}{2} \bigr) \qquad \mbox{with} \quad \frac{\De}{\De t} := \frac\partial{\partial t} + \pounds_v \ , 
\ee
where $s_v= g(v,v)$ is the \emph{speed} (squared) of the fluid. 

The dual formulation of the Euler equation \eqref{eq:Eualt} also holds in the more general settings of Definition~\ref{def:idfluid}, since
\be \label{eq:nablalie}
(\nabla_v v)^\flat = \pounds_v v^\flat - \tfrac{1}{2}\,\de s_v \ ,
\ee
for any Riemannian connection $\nabla.$
\end{remark}

\begin{example}[\textbf{Euclidean Fluid Flows}]
Let $M^n=\IR^n$ with the standard Euclidean metric and volume form. Let $x=(x^1,\dots,x^n)$ be coordinates on $\IR^n$, and abbreviate the corresponding partial derivatives as $\partial_i=\frac{\partial}{\partial x^i}$. Writing the fluid velocity in component form $v = v^i(t,x) \, \frac\partial{\partial x^i}$, the Euler equations \eqref{eq:Euler1} and \eqref{eq:Euler2} reduce respectively to the more familiar equations
\begin{align*}
\frac{\partial v^i}{\partial t} + v^j\,\partial_j v^i = -\partial^i p
\end{align*}
and
\begin{align*}
\partial_i v^i = 0 \ .
\end{align*}
\end{example}

\begin{definition} \label{def:vorticity}
The \emph{vorticity} of the fluid is the multi-vector field $\xi \in \mathsf{\Gamma}_\mu(\midwedge^{n-2}\,TM^n)$ given by the contraction
\be \label{vort1}
\omega_\xi := \xi \,\lrcorner\, \mu = \de v^\flat \ ,
\ee
where $\omega_\xi\in\mathsf{\Omega}^2(M^n)$ is the \emph{vorticity two-form} and $v^\flat\in\mathsf{\Omega}^1(M^n)$ is the \emph{covector potential} of~$\xi.$
\end{definition}

The Euler equations can be reformulated in terms of the vorticity using

\begin{lemma} \label{lemma:arifluid}
The flow equation \eqref{eq:Euler1} can be written in the form
\be  
\frac{\del v}{\del t} = \big(-\iota_v\, \omega_\xi - \de ( p + \tfrac{s_v}{2} ) \big) ^\sharp \ .
\ee
\end{lemma}

\begin{proof}
It follows from Definition \ref{def:vorticity} and the Cartan homotopy formula for the Lie derivative that
\be  
\pounds_v v^\flat = \de\,\iota_vv^\flat + \iota_v\,\de v^\flat = \de s_v + \iota_v\, \omega_\xi  \ .
\ee
The result then follows by substituting this into Equation~\eqref{eq:Eualt}.
\end{proof}

\begin{remark}
It is easy to check that the Euler equations \eqref{eq:Euler1} and \eqref{eq:Euler2}  yield a flow equation for the vorticity field which implies that it is transported exactly by the fluid.
By taking the differential of the dual equation \eqref{eq:Eualt}, we obtain
\be  
\frac\del{\del t} \, \de v^\flat+ \pounds_v\, \de v^\flat = 0 \ ,
\ee
where 
\be  
\pounds_v\, \de v^\flat = \de\, \iota_v\, (\xi\,\lrcorner\, \mu) = (\pounds_v\xi) \,\lrcorner\, \mu \ .
\ee
Hence we obtain
\be  
\frac\del{\del t} \,(\xi\,\lrcorner\, \mu) + (\pounds_v \xi) \,\lrcorner\, \mu =0 \ ,
\ee
and since the fluid is incompressible, i.e. $\pounds_v\mu=0$, this yields the material continuity equation
\be \label{eq:vorticity}
\frac{\De\xi}{\De t} = 0 \ ,
\ee
which is called the \emph{vorticity equation}. Since the fluid density $\mu$ is constant along the fluid lines, the vorticity two-form $\omega_\xi$ is also transported exactly by the fluid flow.
\end{remark} 

\begin{example}[\textbf{Two-Dimensional Fluid Flows}] \label{ex:2d}
In two dimensions, the vorticity $\xi$ is a scalar field. We take $M^2$ to be an oriented surface with $\sfH^1(M^2;\IR) = 0$. 
The divergence-free constraint \eqref{eq:Euler2}, i.e. $\pounds_v\mu=\de\,\iota_v\mu=0$, is then solved by 
\be \label{eq:flux}
\iota_v\mu = \de\psi \ ,
\ee
where $\psi\in C^\infty(M^2)$ is called the \emph{stream function}. The stream function $\psi$ is uniquely defined up to an additive locally constant function, which can be fixed by the requirement $\psi|_{\partial M^2}=0$ when $M^2$ is a connected open domain. The level curves of $\psi$ are called \emph{streamlines}.

By applying the the Hodge $\star_\mu$ operator to Equation \eqref{eq:flux} we find that the fluid velocity is determined by the stream function through
\be
v = \star_\mu\,\de\psi \ .
\ee
Similarly, by applying $\star_\mu$ to Equation \eqref{vort1} one finds that the vorticity of the two-dimensional flow is the function
\be
\xi = \star_\mu \, \de v^\flat = \star_\mu\,\de\,(\star_\mu\,\de\psi)^\flat =: \Delta \psi \ ,
\ee
where $\Delta$ is the Laplacian on $C^\infty(M^2)$.

The vorticity equation \eqref{eq:vorticity} then takes the form of a Hamiltonian equation of motion
\be\label{eq:vorteq2d}
\frac{\partial \xi}{\partial t} = \{\psi,\xi\}_\mu \ ,
\ee
where $\{\psi,\xi\}_\mu := \star_\mu\,(\de\psi\wedge\de\xi)$ is the Poisson bracket on $C^\infty(M^2)$. Hence ideal imcompressible fluid flows in two dimensions can be described entirely in terms of a stream function~$\psi$ playing the role of a Hamiltonian function.
\end{example}

\begin{example}[\textbf{Three-Dimensional Fluid Flows}] \label{rmk:vort}
In three dimensions, from Definition~\ref{def:vorticity} it follows that the vorticity $\xi$ is the divergence-free vector field
\be
\xi = \star_\mu \, \de v^\flat =: \mathrm{curl} (v) \ .
\ee
In other words, the velocity $v$ is the vector potential for $\xi.$ Thus the vorticity describes the local spinning motion, i.e. the tendency of the fluid to rotate, as seen by an observer moving along the fluid flow. 

For $n=3$, vector fields whose interior product with the volume form $\mu$ is an exact differential two-form are called \emph{null homologous}. Null homologous vector fields $\zeta$ allow for the definition of a Hopf invariant $\cH(\zeta)$, called the \emph{helicity}, by using their covector potential to define an abelian Chern-Simons functional.  In particular, the helicity\footnote{Here we assume that integration on the three-manifold is well-defined, for instance this happens when $M^3$ is compact.} of the vorticity $\xi$ is given by
\be  
\cH(\xi) = \int_{M^3} \, v^\flat \wedge \de v^\flat = \int_{M^3} \, g(v, \xi) \ \mu \ .
\ee
Since $v$ is a divergence-free vector field, i.e. it preserves $\mu,$ from Equation~\eqref{eq:Eualt} it follows that the helicity of the vorticity $\cH(\xi)$ is conserved by the fluid:
\be  
\frac{\de \cH(\xi)}{\de t} := \frac{\partial\cH(\xi)}{\partial t} + \pounds_v \cH(\xi) = 0 \ ,
\ee
provided that $\xi$ is parallel to the boundary $\partial M^3$.\footnote{This is of course automatically satisfied when $M^3$ is closed.}

The integral curves of the vorticity vector field $\xi$ are called \emph{vortex lines}. The vorticity equation \eqref{eq:vorticity} becomes
\be
\frac{\partial\xi}{\partial t} = [\xi,v] \ ,
\ee
and it implies that vortex lines flow along fluid lines.
The hydrodynamic helicity $\cH(\xi)$ is an isotopy invariant of the fluid domain $M^3$ which measures the average linking and knotting of vortex lines in the flow.
If the vorticity covector potential $v^\flat$ satisfies the Frobenius integrability condition $v^\flat\wedge\de v^\flat=0$, then 
$\cH(\xi) = 0,$ and hence $g(v, \xi) = 0,$ i.e.~the velocity vector field and the vorticity are orthogonal. 
\end{example}

\medskip

\subsection{Incompressible Fluid Flows with Aristotelian Structure} ~\\[5pt]
We will now study ideal incompressible fluids that flow on oriented $n$-manifolds $M^n$ endowed with an integrable Aristotelian structure. 

\begin{definition} \label{def:arifluid}
An (ideal, incompressible and integrable) \emph{Aristotelian fluid} is given by the data of an ideal incompressible fluid $(\mu,g,v,p)$ flowing in a domain $M^n$ together with a triple $(\tau,h,\gamma)$ such that
\begin{myitemize}
\item $(\tau, v, h, \gamma)$ is a one-parameter family of Aristotelian structures on $M^n,$ where the fluid velocity $v \in \mathsf{\Gamma}_\mu(TM^n)$ is the vector field of observers:
\be
\iota_v \tau =1 \qquad \mbox{and} \qquad \iota_vh=0 \ ;
\ee
\item the subbundles $\ker(\tau)$ and $\mathrm{Span}(v)$ of $TM^n$ are orthogonal with respect to the Riemannian metric $g$ for the fluid flow, where the restriction of the spatial metric $h\in\mathsf{\Gamma}(\midodot^2\, T^*M^n)$ to $\ker(\tau)$ coincides with $g$:
\be
h\,\big|_{\mathsf{\Gamma}(\odot^2\ker(\tau))} = g\,\big|_{\mathsf{\Gamma}(\odot^2\ker(\tau))} \ ;
\ee
\item the restriction of the spatial cometric $\gamma\in\mathsf{\Gamma}(\midodot^2\, TM^n)$ to $\mathrm{Ann}(v)$ is the cometric $g^{-1}$:
\be
\gamma\,\big|_{\mathsf{\Gamma}(\odot^2\mathrm{Ann}(v))} = g^{-1}\big|_{\mathsf{\Gamma}(\odot^2\mathrm{Ann}(v))} \ ;
\ee
and
\item the clock form $\tau\in \mathsf{\Omega}^1(M^n)$ satisfies the Frobenius integrability condition
\be  
\de \tau =  \tau \wedge \alpha \ ,
\ee
for some one-form $\alpha \in \mathsf{\Omega}^1(M^n),$ yielding a one-parameter family of foliations $\cF$ of $M^n$ with $\ker(\tau)= T\cF$. 
\end{myitemize} 

The \emph{configuration space} of an Aristotelian fluid flowing in $M^n$ is the one-parameter family of Lie groups of volume-preserving and foliation-preserving diffeomorphisms $\mathsf{Diff}_\mu(M^n, \cF).$ If $\del M^n\neq\varnothing$, this family of groups also preserves the boundary of $M^n$. 
\end{definition}

Let us unravel and explain the various facets of Definition~\ref{def:arifluid}.
In contrast to the notion of a non-Lorentzian spacetime, in this case time is not determined by the clock form $\tau$. Here $\tau$ is an ingredient introduced to probe the transverse geometry to the fluid lines. Instead, time $t\in\IR$ parametrizes the family of Aristotelian structures $(\tau, v, h, \gamma)$, which evolves according to the Euler equations. Because the velocity vector field $v$ is required to be nowhere-vanishing, Aristotelian fluid lines are determined by flow equations that have no fixed points.

At each fixed time, the Aristotelian structure allows for an orthogonal decomposition of the tangent bundle of the fluid domain as $TM^n = \ker(\tau) \oplus \textrm{Span}(v)$, which yields a choice of frame adapted to the fluid flow. The background Riemannian metric $g$ can be written with respect to this frame as
\be \label{eq:metricfluid}
g = h + \tfrac{s_v}2 \, \tau \otimes \tau  \ ,
\ee
where the fluid speed $s_v = g(v,v) \in C^\infty(M^n)$ is a positive function.
Since $g$ is independent of time, the flow of the spatial metric $h$ is given by
\be
\frac{\partial h}{\partial t} = -\frac12 \, \frac{\partial s_v}{\partial t} \, \tau\otimes\tau - s_v \, \tau\odot \frac{\partial\tau}{\partial t} \ .
\ee
In Section~\ref{subs:TransportEqs} we will derive explicit flow equations for the speed $s_v$ and clock form $\tau$ of an Aristotelian fluid.

The spatial cometric $\gamma$ will actually play no direct role in the following, and in principle it could be left arbitrary, subject only to the defining property $\iota_\tau\gamma=0$ of an Aristotelian structure. For definiteness we have taken it to be dual to the spatial metric $h$, i.e. $h(\gamma) = 1$, which agrees with the usual conventions taken in the literature. For each fixed time, the Aristotelian structure yields a splitting of the cotangent bundle $T^*M^n = \mathrm{Ann}(v)\,\oplus\,\mathrm{Span}(\tau)$, which gives a decomposition of the cometric $g^{-1}$ in this coframe as
\be \label{eq:cometricfluid}
g^{-1} = \gamma + \tfrac1{2\,s_v} \, v\otimes v \ ,
\ee
with $\frac1{s_v}=g^{-1}(\tau,\tau)$. The flow of $\gamma$ is given by
\be
\frac{\partial \gamma}{\partial t} = \frac1{2\,s_v^2} \, \frac{\partial s_v}{\partial t} \, v\otimes v - \frac1{s_v} \, v\odot\frac{\partial v}{\partial t} \ .
\ee

The background volume form $\mu$ characterising the density of the fluid is directly related to the family of Aristotelian structures. It is given by a reduction of the $\sfG_\ttA$-structure to the component $\sfG_{\ttA\,0}\simeq\sfS\sfO(n-1)$ of $\sfG_\ttA\simeq \sfO(n-1)$ which is connected to the identity. 

We include the integrability condition in Definition~\ref{def:arifluid} because it represents a crucial ingredient used in this paper and it simplifies some of our analysis in the following. It is not needed for the description of either Aristotelian structures or fluid dynamics, but it does represent an important class of physically relevant cases.
As discussed in Remark~\ref{rmk:transvfol}, the transverse component of the metric \eqref{eq:metricfluid}, together with $\cF,$ does not define a Riemannian foliation. The leaves of $\cF$ can be interpreted as sections of the fluid orthogonal to the fluid lines with respect to the metric $g.$ 

Finally, the definition of the configuration space of an Aristotelian fluid can be motivated in the following way. It is shown in \cite{Moerdijk2003} that any transversely parallelisable foliation $\cF$ on a compact connected $n$-manifold $M^n$ is \emph{homogeneous}, i.e. for any $x, \, y \in M^n,$ there exists a diffeomorphism $\varphi \in \mathsf{Diff}(M^n)$ preserving the foliation $\cF$ such that $\varphi(x) = y \, .$
This is consistent with Definition~\ref{def:idfluid}, since these diffeomorphisms define the fluid lines characterising the flow of the ideal fluid. 
Since the foliation is homogeneous, the group $\mathsf{Diff}_\mu(M^n, \cF)$ acts transitively on $M^n.$ This is the first step towards a variational formulation of Aristotelian fluid flows in terms of geodesic equations on the configuration space $\mathsf{Diff}_\mu(M^n, \cF)$, which can be achieved following the approach of~\cite{Arnold}.

\begin{remark}
Let us compare Definition~\ref{def:arifluid} with two other relevant approaches in the literature based on non-Lorentzian geometry:
\begin{myitemize}
\item 
Our definition of an Aristotelian fluid is analogous to the ideal non-boost invariant fluids discussed in \cite{deBoer:2020xlc}, except that we explicitly break the local rotational symmetry to the subgroup $\sfG_{\ttA\,0}\simeq\sfS\sfO(n-1)$ preserving the fluid lines. Our geometric approach to fluid dynamics is different because we work directly with the hydrodynamic equations themselves, rather than deriving them as a non-Lorentzian limit of the conservation laws for the energy-momentum tensor of a general relativistic fluid. The latter is also discussed in \cite{Petkou:2022bmz}, where this point of view is complemented by deriving the conservation laws for the fluid from diffeomorphism invariance of the given spacetime structure.   
\item
A three-dimensional fluid flow endowed with a Carrollian geometry is discussed in \cite{Freidel2022}, where the spacetime structure is determined by the Hopf fibration of $S^3$ (or more generally any homology three-sphere) with the fluid velocity similarly identified as the observer vector field of the Carrollian structure. Our approach is inspired by this construction and may be regarded as an Aristotelian counterpart in arbitrary backgrounds and dimensions. However, a Carrollian fluid is also different from our notion of Aristotelian fluid (for $n=3$), since it has a natural codimension two foliation given by the fibres of the Hopf map $S^3\to S^2$ that determines the velocity vector field, whereas the spatial metric is obtained as the pullback of any metric on $S^2.$
\end{myitemize}
\end{remark}

The topology of an Aristotelian fluid flow may be determined by applying Theorem~\ref{thm:Globalreeb} and Corollary~\ref{cor:Reeb}. It is moreover always true in this construction that the hypersurfaces orthogonal to the fluid lines with respect to the metric $g,$ given by the leaves of the integrable Aristotelian structure, have trivial holonomy (see Remark \ref{rmk:transvfol}).

\begin{remark} \label{rmk:GVtauv}
Since
\be \label{eq:vtau}
v^\flat = \iota_v g = s_v\, \tau \ ,
\ee
by Lemma \ref{LemmaTau2} it follows that the one-forms $\tau$ and $v^\flat$  define the same foliation $\cF$ with 
\be \label{eq:valpha}
\de v^\flat = v^\flat \wedge \alpha_v \ ,
\ee 
where 
\be \label{eq:alphataualpha}
\alpha_v= \alpha - \de \log s_v \ .
\ee
\end{remark}

\begin{remark} \label{rmk:arimetric}
When written in terms of fluid variables using the adapted frame in Equation~\ref{eq:metricfluid}, the metric $g$ is not only expressed in terms of the Aristotelian structure but also in terms of a  positive function $s_v$ which determines the speed of the fluid. This  greatly affects the geometric characterisation of the fluid flow, and different speed functions $s_v$ correspond to different states of the fluid. 

For instance, the one-form $\alpha_v$ can be expressed in terms of the quantities characterising our fluid and depends on $s_v$ as well. From Equation~\eqref{vort1} it follows that the vorticity two-form is given by
\be  \label{eq:muvflat}
\omega_\xi = v^\flat \wedge \alpha_v \ ,
\ee
which yields
\be \label{eq:alphafluid}
\alpha_v = \tfrac{1}{s_v}\, \iota_v\,\omega_\xi + \tfrac{1}{s_v}\, (\iota_v \alpha_v)\, v^\flat \ .
\ee
On the other hand
\be  \label{eq:alphavf}
\alpha = \pounds_v \tau = \tfrac{1}{s_v}\,  \iota_v\,\omega_\xi - (\pounds_v \log s_v)\, \tau + \de \log s_v \ ,
\ee
where we used Equation~\eqref{eq:vtau}. It then follows from Equation~\eqref{eq:alphataualpha} that
\be \label{eq:alphav}
\iota_v \alpha_v = - \pounds_v \log s_v = - \tfrac{1}{s_v}\, \pounds_v s_v \ .
\ee

From this calculation it also follows that $\alpha_v$ is completely tangential to the foliation $\cF$:
\be \label{eq:alphatang}
\iota_v \alpha_v = 0 \ ,
\ee
if and only if
\be \label{eq:transvelocity}
\pounds_v s_v = 0 \ .
\ee
Equivalently, the condition \eqref{eq:transvelocity} implies that
\be \label{eq:alphatransv}
\alpha_v = \tfrac{1}{s_v}\, \iota_v\, \omega_\xi \ \in \ \mathsf{\Gamma}({\mathrm{Ann}}(v)) \ .
\ee
This means that the speed of the fluid can only change along the leaves of the foliation which are transversal to the fluid lines. 
\end{remark}

\begin{remark} \label{rmk:eulerisable}
Under suitable conditions, an Aristotelian fluid provides an example of an \emph{Eulerisable flow}, see e.g.~\cite{Rechtman:2020}. Let $(M^n, \mu)$ be an oriented $n$-manifold. Then a volume-preserving vector field $v \in \mathsf{\Gamma}_\mu(TM^n)$ is \emph{Eulerisable} if
there exists a Riemannian metric $g$ on $M^n$ such that
\be  \label{eq:Eulerisablev}
\iota_v\, \de v^\flat = \iota_v\,\omega_\xi = \de b \ ,
\ee
for some function $b \in C^\infty(M^n)$ called the \emph{Bernoulli function}. 

By Lemma~\ref{lemma:arifluid}, an Eulerisable vector field yields a stationary or steady solution of the Euler equations, i.e. $\frac{\partial v}{\partial t}=0$, by fixing the pressure field~$p$ in terms of $b$ and $s_v$. By definition, the Bernoulli function is constant along the fluid lines: 
\be
\pounds_v b = 0 \ ,
\ee
or equivalently $\de b\in\mathsf{\Gamma}(\mathrm{Ann}(v))$.
Hence an Eulerisable flow is possible only if the renormalised pressure $p+\frac{s_v}2$ is also constant along the fluid lines. Note that by the vorticity equation \eqref{eq:vorticity}, the vorticity $\xi$ of any steady flow is constant along the fluid lines.

For the Aristotelian fluid flow of Definition \ref{def:arifluid}, where $\mathrm{Ann}(v)=T^*\cF$, the velocity vector field $v$ is Eulerisable with respect to the metric $g$ defining the Aristotelian structure if and only if
\be \label{eq:arieulerisable}
(\pounds_v s_v) \, \tau + s_v \, \alpha_v = \de b \ ,
\ee
where we used Equations \eqref{eq:vtau} and \eqref{eq:valpha}. In other words, we check if the metric $g$ defining the Aristotelian structure satisfies Equation \eqref{eq:Eulerisablev}, i.e. $v$ is a stationary flow for that given metric. We may refer to this case as \emph{stationary} (or \emph{steady}) \emph{Aristotelian flow with a Bernoulli function}.
\end{remark}

\medskip

\subsection{Transport Equations} ~\\[5pt] \label{subs:TransportEqs}
We will now derive some transport equations which will prove useful in the following. In particular, since the speed of the fluid $s_v$ evidently plays a prominent role in our approach, let us determine its flow along the fluid lines. In the general case, we obtain

\begin{proposition} \label{prop:transpf}
The speed $s_v= g(v,v)$ of any ideal incompressible fluid $(M^n,\mu,g,v,p)$ obeys the transport equation 
\be  
\frac{\De s_v}{\De t}  = - 2\, \pounds_v \big(p+ \tfrac{s_v}{2}\big) \ .
\ee
\end{proposition}

\begin{proof}
By using the flow equation \eqref{eq:Euler1}, we find
\begin{align}
\frac{\del s_v}{\del t} = \iota_{\frac{\del v}{\del t}} v^\flat + \iota_v\, \frac{\del v^\flat}{\del t} = - g (\nabla_v v, v ) - \pounds_v p - \iota_v\, \pounds_v v^\flat - \pounds_v p - \tfrac{1}{2}\, \pounds_v s_v \ .  
\end{align}
From Equation \eqref{eq:nablalie}, together with $2\, g(\nabla_v v, v)= \pounds_v s_v,$ we obtain
\be  
\frac{\del s_v}{\del t} = - 2\, \pounds_v s_v - 2\, \pounds_v p
\ee
and the result follows.
\end{proof}

Proposition~\ref{prop:transpf} yields a constraint on the pressure $p$ if the speed $s_v$ is constant: $\pounds_v p = 0.$ Hence for fluid flows of constant speed, the pressure field must be constant in the direction of the fluid lines, i.e. $\de p\in\mathsf{\Gamma}(\mathrm{Ann}(v))$. In particular, this affects the geometry of an Aristotelian fluid, in which the decomposition of the metric $g$ in the frame adapted to the fluid flow from Equation~\eqref{eq:metricfluid} is determined solely by the spatial metric $h$ and the clock form $\tau$. For an Eulerisable fluid flow this is automatically satisfied because it is simply a property of the Bernoulli function $b$.

We can now determine the transport equation for the clock form $\tau$ of the Aristotelian fluid. We find

\begin{proposition} \label{prop:transptau}
The clock form $\tau$ of an Aristotelian fluid $(M^n, \mu,\tau, v , g,p)$ satisfies the transport equation
\be \label{eq:transtau}
\frac{\De\tau}{\De t} = \frac1{s_v}\,\Big(2\,\pounds_v \big(p+\tfrac{s_v}{2}\big)\, \tau -  \de \big(p-\tfrac{s_v}{2} \big)\Big) \ .
\ee
\end{proposition}

\begin{proof}
From $\tau = \frac{1}{s_v}\, v^\flat$ and Lemma~\ref{lemma:arifluid} we obtain
\begin{align}
\frac{\del\tau}{\del t} =& - \frac{1}{s_v} \, \frac{\del s_v}{\del t} \, \tau + \frac{1}{s_v}\, \frac{\del v^\flat}{\del t} = - \frac{1}{s_v} \, \frac{\del s_v}{\del t}  \, \tau - \frac{1}{s_v} \, \iota_v\, \omega_\xi - \frac{1}{s_v}\, \de \big(p + \tfrac{s_v}{2}\big) \ .  
\end{align}
We recall that
\be  
\tfrac{1}{s_v}\, \iota_v\, \omega_\xi = \pounds_v \tau - \de \log s_v + (\pounds_v \log s_v)\, \tau \ .
\ee
Thus
\be  
\frac{\De \tau}{\De t} = - \frac1{s_v}\,\frac{\De s_v}{\De t}  \,\tau -\frac{1}{s_v}\, \de \big(p- \tfrac{s_v}{2} \big) \ .
\ee
By using Proposition \ref{prop:transpf}, we obtain Equation~\eqref{eq:transtau}.
\end{proof}

From the transport equation \eqref{eq:transtau} it follows that if the fluid speed $s_v$ is constant, then $\tau$ is a locally conserved covector field.

Lastly we consider the transport equation for $\alpha_v$, which is given by

\begin{proposition}
Let $(M^n, \mu,\tau, v , g,p)$ be an Aristotelian fluid. Then the flow equation for the one-form $\alpha_v$ defined by Equation~\eqref{eq:valpha} is
\begin{align} \label{eq:evalpha}
\begin{split}
\frac{\De\alpha_v}{\De t} =& \, - \frac{1}{s_v}\, (\pounds_v\log s_v) \, \iota_v\, \omega_\xi +\frac{1}{s_v}\, \pounds_v \big(p + \tfrac{5\,s_v}{2}\big) \, \alpha_v \\ 
& \quad - \kappa \, v^\flat + \frac1{s_v}\,(\pounds_v \log s_v) \, \de \big(p- \tfrac{s_v}{2}\big) \ , 
\end{split}
\end{align}
where $\kappa$ is the function
\be \label{eq:kappatransport}
\kappa = \tfrac{1}{s_v}\,\kappa^\prime -\tfrac{3}{s_v}\,(\pounds_v \log s_v)^2 +  \tfrac{1}{s_v}\, \pounds^2_v \log s_v
\ee
and
\begin{align}  
\begin{split}
\kappa^\prime &= -s_v \, g^{-1}(\alpha_v, \alpha_v) - g^{-1}\big(\de (p + \tfrac{s_v}{2}), \alpha_v\big) - \tfrac1{s_v} \, g^{-1}\big(\iota_v\, \omega_\xi - \de (p+\tfrac{s_v}2) , \de s_v \big) \\
& \qquad - \tfrac2{s_v}\,\pounds^2_v(p + s_v) + \tfrac2{s_v} \, (\pounds_v\log s_v)\,\pounds(p+\tfrac{s_v}2)  \ .
\end{split}
\end{align}
\end{proposition}

\begin{proof}
Let us take the differential of the dual of the Euler equation from Lemma \ref{lemma:arifluid}:
\begin{align}  
 \frac\del{\del t}\, \de v^\flat = - \de \,\iota_v\, \omega_\xi
\end{align}
and use Equation \eqref{eq:valpha} to obtain 
\begin{align} \label{eq:transportp1}
 \frac{\del v^\flat}{\del t} \wedge \alpha_v + v^\flat \wedge \frac{\del\alpha_v}{\del t} = - \de\,\iota_v\, \omega_\xi  \ .
\end{align}
By using Equation \eqref{eq:alphafluid}, we can express the dual of Lemma \ref{lemma:arifluid} as
\begin{align} \label{eq:dualLemma}
 \frac{\del v^\flat}{\del t} =  -  s_v\, \alpha_v - (\pounds_v \log s_v )\,v^\flat - \de(p + \tfrac{s_v}{2}) 
\end{align}
and compute
\begin{align} \label{eq:transportp2}
 \frac{\del v^\flat}{\del t} \wedge \alpha_v = - (\pounds_v \log s_v )\,\de v^\flat - \de (p + \tfrac{s_v}{2}) \wedge \alpha_v  \ . 
\end{align}
By substituting Equation \eqref{eq:transportp2} in Equation \eqref{eq:transportp1} we obtain
\begin{align} \label{eq:transportp3}
 v^\flat \wedge \frac{\del\alpha_v}{\del t} = - \de\,\iota_v\, \omega_\xi + (\pounds_v \log s_v)\, \de v^\flat + \de(p + \tfrac{s_v}{2}) \wedge \alpha_v \ .
\end{align}
By taking the interior product $\iota_v$ on both sides of Equation \eqref{eq:transportp3} we get
\begin{align} \label{eq:transportp4}
\begin{split}
\hspace{-5mm} s_v\, \frac{\del\alpha_v}{\del t} = & \ \Big(\iota_v\, \frac{\del\alpha_v}{\del t}\Big) \, v^\flat - \pounds_v \,\iota_v\, \omega_\xi + (\pounds_v \log s_v)\, \bigr( s_v\, \alpha_v + (\pounds_v \log s_v)\, v^\flat \bigl)\\
 & + \pounds_v(p + \tfrac{s_v}{2})\, \alpha_v + (\pounds_v \log s_v)\, \de (p + \tfrac{s_v}{2}) \ ,
 \end{split}
\end{align}
where we used Equation \eqref{eq:alphav}.

We first show that the term $\iota_v\, \frac{\del\alpha_v}{\del t}$ gives a function in which the flow of $\alpha_v$ plays no role.
From the Leibniz rule for the time derivative operator, we find 
\be  
\iota_v\, \frac{\del\alpha_v}{\del t} = - \iota_{\frac{\del v}{\del t}} \alpha_v - \frac\del{\del t} \, \pounds_v \log s_v \ ,
\ee
where
\be  
\frac\del{\del t} \, \pounds_v \log s_v= - \frac1{s_v^2}\, (\pounds_v s_v)\, \frac{\del s_v}{\del t} + \frac{1}{s_v}\,\frac\del{\del t} \, \pounds_v s_v \ .
\ee
By Proposition~\ref{prop:transpf}, the transport equation for $s_v$ yields
\be  
\pounds_v\, \frac{\del s_v}{\del t} = - 2\,\pounds^2_v( p + s_v) \ ,
\ee
and hence
\be  
\frac\del{\del t} \, \pounds_v s_v = \iota_{\frac{\del v}{\del t}}\,\de s_v - 2\, \pounds^2_v(p +s_v)  \ ,
\ee
where Lemma \ref{lemma:arifluid} yields
\begin{align}  
\iota_{\frac{\del v}{\del t}}\,\de s_v = -g^{-1}\big(\iota_v\, \omega_\xi + \de( p + \tfrac{s_v}2) , \de s_v\big) \ .
\end{align}
By using the dual of Equation \eqref{eq:dualLemma}, we can easily obtain
\be  
 \iota_{\frac{\del v}{\del t}} \alpha_v = -(\pounds_v \log s_v)^2 - \kappa_{\alpha_v}  \ ,
\ee
where
\be  
\kappa_{\alpha_v} = s_v\, g^{-1}(\alpha_v, \alpha_v) + g^{-1}\big(\de ( p +\tfrac{s_v}{2}), \alpha_v\big) \ .
\ee
Putting these calculations together, we therefore get
\begin{align} \label{eq:iotadelalpha}
\begin{split}
\iota_v\, \frac{\del\alpha_v}{\del t} &= \kappa_{\alpha_v} + \frac1{s_v}\,g^{-1}\big(\iota_v\, \omega_\xi + \de( p + \tfrac{s_v}2), \de s_v\big) \\
& \quad \, + \frac2{s_v}\, \pounds^2_v(p +s_v) - \frac2{s_v}\,(\pounds_v\log s_v)\,\pounds_v(p+\tfrac{s_v}2) \ .   
\end{split}
\end{align}

To write the standard form of a transport equation, we need to show how the Lie derivative $\pounds_v \alpha_v = \de\,\iota_v\alpha_v + \iota_v\,\de\alpha_v$ appears in Equation \eqref{eq:transportp4}.
From Equation~\ref{eq:alphav} we find
\be \label{eq:transportpp}
\de\, \iota_v \alpha_v= - \de\, \pounds_v \log s_v \ ,
\ee
while from Equation \eqref{eq:alphafluid} we get
\begin{align} \label{eq:transportp6}
\begin{split}
\iota_v \, \de \alpha_v  = & \ \iota_v \, \de \bigl(\tfrac{1}{s_v} \, \iota_v\, \omega_\xi - \tfrac1{s_v}\,(\pounds_v \log s_v)\,v^\flat \bigr) \\[4pt]
= & - \tfrac{1}{s_v}\,(\pounds_v \log s_v) \, \iota_v\, \omega_\xi - \tfrac{1}{s_v}\, \pounds_v \,\iota_v\, \omega_\xi - \pounds_v \bigl(  \tfrac{1}{s_v}\pounds_v \log s_v  \bigr)\, v^\flat\\ 
& \quad + s_v \, \de \bigl( \tfrac{1}{s_v} \pounds_v \log s_v  \bigr) + \tfrac{1}{s_v} \, (\pounds_v \log s_v) \, \big(s_v\,\alpha_v+(\pounds_v\log s_v)\, v^\flat\big) \ .
\end{split}
\end{align}
We now add $s_v \, \pounds_v \alpha_v$ to both sides of Equation \eqref{eq:transportp4}, substituting Equations \eqref{eq:transportpp} and \eqref{eq:transportp6} on the right-hand side. After a little algebra, we get 
\begin{align} \label{eq:transportppp} 
\begin{split}
\frac{\De \alpha_v}{\De t} = & - \kappa \, v^\flat - \frac{}{s_v}\,(\pounds_v\log s_v)\,\iota_v\, \omega_\xi + \frac{1}{s_v} \pounds_v\big(p + \tfrac{5\,s_v}{2}\big)\, \alpha_v \\ 
&+ s_v \,\de \bigl(\tfrac{1}{s_v}\, \pounds_v \log s_v \bigr) - \de\, \pounds_v \log s_v + \frac1{s_v}\, (\pounds_v \log s_v)\, \de (p + \tfrac{s_v}{2}) \ , 
\end{split}
\end{align}
where
\begin{align}  
\kappa = -\frac{1}{s_v}\, \iota_v\, \frac{\del\alpha_v}{\del t} +  \frac{1}{s_v}\, \pounds_v^2 \log s_v - \frac{3}{s_v}\, (\pounds_v \log s_v)^2 \ ,
\end{align}
which is easily shown to take the form \eqref{eq:kappatransport} by using Equation \eqref{eq:iotadelalpha}. Using
\begin{align}
s_v\, \de \bigl( \tfrac{1}{s_v}\, \pounds_v \log s_v\bigr) - \de \, \pounds_v \log s_v = - (\pounds_v \log s_v)\, ( \de \log s_v ) \ ,  
\end{align}
we obtain Equation \eqref{eq:evalpha} from Equation \eqref{eq:transportppp}.
\end{proof}

\medskip

\subsection{Torsion of Aristotelian Fluid Flows} ~\\[5pt]
Let $\nabla^\ttA$ be a compatible Aristotelian connection for the structure tensors $(\tau,v,h,\gamma)$ of an Aristotelian fluid with foliation $\cF$. Note that $\nabla^\ttA$ is not a Riemannian connection for $g$ unless \smash{$\de^{\nabla^\ttA} s_v  = 0.$} We will now show that the torsion tensor \smash{$T^{\nabla^\ttA}$} of $\nabla^\ttA$ is completely determined by the quantities characterising the fluid. This a rare instance in which the torsion can be computed in such an explicit form.

It follows from Proposition~\ref{prop:Zmu} and from $\pounds_v \mu = 0$ that
\smash{$\mathrm{tr}\big(\iota_v T^{\nabla^\ttA}\big) = 0$}.
As before, we write \smash{$\bar{T}^{\nabla^\ttA} = \tau\circ T^{\nabla^\ttA}$}. According to the classification discussed in Section \ref{subs:Intrinsic}, for an Aristotelian spacetime admitting a foliation with $\de \tau \neq 0$ and vector field $v$ of observers preserving the volume form $\mu$, it follows from Equation~\eqref{eq:classmuT} that the torsion tensor satisfies
\be \label{eq:ivTAnn}
\iota_v \bar{T}^{\nabla^\ttA} \ \in \ \mathsf{\Gamma}(\mathrm{Ann}(v))= \mathsf{\Gamma}(T^*\cF) \ .
\ee

We can easily demonstrate that this is the case for ideal incompressible fluids: Recall from Equation~\eqref{eq:alpT} that
\be  
\iota_v  \bar{T}^{\nabla^\ttA} = \alpha \ ,
\ee
where in the present case the right-hand side is determined by
\be \label{eq:alphat}
\alpha = \tfrac{1}{s_v}\, \iota_v\, \omega_\xi - \tfrac{1}{s_v}\, (\pounds_v \log s_v)\,v^\flat + \de \log s_v \ .
\ee
Then a straightforward calculation gives $\iota_v \alpha = 0,$ as expected. Note that the ``gauge transformation'' \eqref{eq:alphataualpha} preserves the annihilator $\mathrm{Ann}(v)$ of the fluid velocity if and only if the fluid speed $s_v$ is constant along the fluid lines.

The torsion tensor of the fluid is thus completely determined by the clock form, the vorticity and the fluid velocity as the two-form
\be  \label{eq:torsionAri}
\bar{T}^{\nabla^\ttA}= \tfrac{1}{s_v}\, \tau \wedge \pounds_v v^\flat = \tfrac{1}{s_v} \, \tau \wedge \big(\iota_v\, \omega_\xi + \de s_v\big) \ ,
\ee
where we used Proposition \ref{prop:Ztau} together with Equation~\eqref{eq:alphat}. 

\begin{remark} \label{rmk:inttorsionArifluid}
Equation \eqref{eq:torsionAri} provides a simple criterion for integrability of the underlying $\sfS \sfO(n-1)$-structure (see Remark \ref{rmk:integrGAri}), which is interpreted as the existence (at first order) of an $\sfS\sfO(n-1)$-frame moving along the fluid flow. From Equation~\eqref{eq:torsionAri} it follows that the torsion tensor $\bar T^{\nabla^\ttA}$ vanishes if and only if
\begin{align} \label{eq:integrabilityOA}
 \pounds_v v^\flat = k \, v^\flat \ ,   
\end{align}
for some function $k \in C^\infty(M^n).$ 
Equation \eqref{eq:integrabilityOA} is satisfied, for instance, when the fluid velocity $v$ is a conformal Killing vector field of the background Riemannian metric $g$.
The condition \eqref{eq:integrabilityOA} can be easily rewritten in terms of the velocity and the vorticity  as
\begin{align}
\iota_v\, \omega_\xi  + \de s_v = k \, v^\flat \ .
\end{align}

For Aristotelian fluid flows with $\pounds_v s_v=0,$ the only solution is $k=0$, and $v$ is a Killing vector field for $g.$  From Equation \eqref{eq:alphat} it then follows that $\alpha = 0$, 
and hence $\de \tau = 0.$ This  implies $\pounds_v\tau=0$, which is consistent with the decomposition \eqref{eq:metricfluid}.
\end{remark}

The transport equation \eqref{eq:evalpha} is nothing but the flow equation for the tensor \smash{$\iota_v \bar{T}^{\nabla^\ttA}$} characterised by the torsion of the underlying Aristotelian structure, up to exact terms. More generally, we have

\begin{proposition}
 The torsion tensor $\bar{T}^{\nabla^\ttA}$ of an Aristotelian fluid $(M^n, \mu,\tau, v , g,p)$ satisfies the transport equation
 \begin{align} \label{eq:transpTors}
  \frac{\De \bar{T}^{\nabla^\ttA}}{\De t} = \frac{3}{s_v}\,\pounds_v \big(p + \tfrac{s_v}{2}\big) \, \bar{T}^{\nabla^\ttA} + \frac{1}{s_v}\, \tau \wedge \de \bigl( \pounds_v (p + \tfrac{s_v}{2}) - \log s_v  \bigr) \ .
 \end{align}
\end{proposition}

\begin{proof}
 We apply the material time derivative operator $\tfrac{\De}{\De t}$ to  Equation~\eqref{eq:torsionAri}. Using the Leibniz rule, we get
 \begin{align} \label{eq:leibnizTorsion}
 \begin{split}
 \frac{\De \bar{T}^{\nabla^\ttA}}{\De t} &= -\frac{1}{s_v}\, \bigg( \frac{\De s_v}{\De t} \, \bar{T}^{\nabla^\ttA} -  \frac{\De \tau}{\De t} \wedge \big(\iota_v\, \omega_\xi - \de s_v\big) \\
 & \hspace{2cm} + \tau \wedge \Bigl((-1)^{n-1} \, \xi\,\lrcorner\, \frac{\De}{\De t}\,\iota_v\mu - \de\, \frac{\De s_v}{\De t} \Bigr)    \bigg) \ ,
 \end{split}
 \end{align}
 where we used the vorticity equation \eqref{eq:vorticity} in the last term.
We also have
\begin{align}
 \frac{\De}{\De t} \, \iota_v\mu = \iota_{\frac{\del v}{\del t}} \mu \ ,   
\end{align}
and hence
\begin{align}
(-1)^{n-1} \, \xi\,\lrcorner\, \frac{\De}{\De t}\,\iota_v\mu = - \iota_{\frac{\del v}{\del t}} \big(v^\flat \wedge \alpha_v\big)  
 =  \pounds_v \big(p + \tfrac{s_v}{2}\big)\, \alpha_v + \big(\iota_{\frac{\del v}{\del t}} \alpha_v \big)\, v^\flat \ ,
\end{align}
where we used Equation \eqref{eq:muvflat} for the first equality and Lemma \eqref{lemma:arifluid} for the first term in the second equality. By combining this with Equations \eqref{eq:alphataualpha} and \eqref{eq:Talph}, we get
\begin{align} \label{eq:omegavflow}
(-1)^{n-1} \, \tau \wedge \Big( \xi\,\lrcorner\, \frac{\De}{\De t}\,\iota_v\mu \Big) =  \pounds_v \big(p + \tfrac{s_v}{2}\big)\, \big(\bar{T}^{\nabla^\ttA} - \tau \wedge \de \log s_v\big) \ .
\end{align}
By substituting Equation \eqref{eq:omegavflow} in Equation \eqref{eq:leibnizTorsion}, and using Propositions~\ref{prop:transpf} and \ref{prop:transptau}, we obtain the transport equation \eqref{eq:transpTors}.
\end{proof}

When the fluid speed $s_v$ is constant, and hence the pressure field $p$ is constant along the fluid lines, Equation~\eqref{eq:transpTors} implies that the torsion is transported exactly by the fluid flow. This is consistent with the discussion in Remark~\ref{rmk:inttorsionArifluid}.

\medskip

\subsection{Two-Dimensional Aristotelian Fluid Flows} ~\\[5pt]
To exhibit some concrete examples and physical features at this stage, let us momentarily focus on Aristotelian fluids in two spatial dimensions. Let $M^2$ be an oriented surface with $\sfH^1(M^2;\IR)=0$. In this case, there are two special simplifying features that do not appear in higher dimensions. Firstly, any nowhere-vanishing one-form $\tau\in\mathsf{\Omega}^1(M^2)$ is automatically integrable. Secondly, the flow is described by a stream function $\psi\in C^\infty(M^2)$, see Example~\ref{ex:2d}. In the following we will rewrite all the data and equations for two-dimensional Aristotelian fluid flows in terms of $\psi$.

Recall that the vorticity is the function $\xi=\Delta\psi$, where the stream function $\psi$ is defined by the one-form $\iota_v\mu=\de\psi$. For an Aristotelian flow, whose velocity $v$ is a nowhere-vanishing vector field, this requires that the stream function $\psi$ have no critical points\footnote{Note that this condition prevents $M^2$ from being compact, since smooth functions on compact surfaces always have a critical point.} on $M^2$. From $v = \star_\mu\,\de\psi$ we find the clock from
\be \label{eq:clockstream}
\tau = \tfrac1{s_v} \, (\star_\mu\,\de\psi)^\flat \ ,
\ee
whose kernel $\ker(\tau) = T\cF$ gives a one-dimensional foliation $\cF$ of the domain $M^2$. From $\iota_v\tau=1$ we obtain the fluid speed
\be
s_v = \iota_v\,(\star_\mu\,\de\psi)^\flat \ ,
\ee
and the decomposition of the background Riemannian metric $g$ in the frame adapted to the Aristotelian fluid flow is
\be
g = h \, + \, \tfrac1{2\,s_v^2} \, (\star_\mu\,\de\psi)^\flat \, \otimes \, (\star_\mu\,\de\psi)^\flat  \ ,
\ee
where $h$ is the restriction of $g$ to $\ker\big((\star_\mu\,\de\psi)^\flat\big)=T\cF$. 

The one-form $\alpha$ appearing in the integrability condition $\de \tau = \tau\wedge\alpha$ is given from \eqref{eq:alphavf} by
\be
\alpha = \tfrac1{s_v} \,  \Delta\psi \ \de\psi - \tfrac1{s_v} \, \big(\iota_v\,\de\,\iota_v\,(\star_\mu \de\psi)^\flat\big) \, (\star_\mu\,\de\psi)^\flat + \tfrac1{s_v} \, \de\,\iota_v\,(\star_\mu\,\de\psi)^\flat \ .
\ee
The torsion of the two-dimensional fluid is given from \eqref{eq:torsionAri} by the two-form
\be
\bar{T}^{\nabla^\ttA} = \tfrac1{s_v^2}\, (\star_\mu\, \de\psi )^\flat \wedge \big(\Delta\psi \ \de\psi + \de\,\iota_v\,(\star_\mu\,\de\psi)^\flat\big) \ .
\ee

When $\psi$ is independent of time, the condition for a steady two-dimensional flow with a Bernoulli function reads
\be
\Delta\psi \ \de\psi = \de b \ .
\ee
This determines the Bernoulli function $b\in C^\infty(M^2)$ from the stream function, uniquely up to a locally constant function on $M^2$.

\begin{example}[\textbf{Euclidean Fluid Flows}] \label{ex:2dEucl}
We show that incompressible flows on open domains $M^2\subseteq\IR^2$ naturally have the structure of a two-dimensional Aristotelian fluid, wherein the formulas simplify to explicit expressions in terms of the stream function $\psi$ and its gradients. We denote coordinates of $\IR^2$ as $(x,y)$ and the corresponding partial derivatives as $(\partial_x,\partial_y)$, with the standard Euclidean metric and measure
\be
g = \de x\otimes \de x + \de y\otimes \de y \qquad \mbox{and} \qquad \mu = \de x \wedge \de y \ .
\ee
This induces the standard two-dimensional Euclidean Laplace operator
\be
\Delta = \partial_x^2 + \partial_y^2 \ .
\ee

The fluid velocity and speed in this case are given by
\be
v = - \partial_y\psi\,\tfrac\partial{\partial x} + \partial_x\psi\,\tfrac\partial{\partial y} \qquad \mbox{and} \qquad s_v = (\partial_x\psi)^2 + (\partial_y\psi)^2 \ .
\ee
The annihilator of the velocity vector field is given by
\be
\mathrm{Ann}(v) = \mathrm{Span}(\de\psi) \ .
\ee

For the clock form of the Aristotelian structure we find
\be
\tau = - \frac{\partial_y\psi \, \de x - \partial_x\psi \, \de y}{(\partial_x\psi)^2 + (\partial_y\psi)^2} \ .
\ee
Its kernel $\ker(\tau) = T\cF$ defines the one-parameter family of foliations $\cF$ of the fluid domain $M^2$ given by the integrable distribution
\be
T\cF = \mathrm{Span}\big((\de\psi)^\sharp\big) = \mathrm{Span}\big(\partial_x\psi\,\tfrac\partial{\partial x} + \partial_y\psi\,\tfrac\partial{\partial y}\big)
\ee
which is orthogonal to
$\mathrm{Span}(v)$
with respect to the Euclidean metric $g$. Its leaves $L_r(t)$ for $r\in\IR$ are just the {streamlines} of the flow:
\be
L_r(t) = \big\{(x,y)\in M^2 \ \big| \ \psi(t,x,y) = r \big\} \ .
\ee
The spatial metric of the Aristotelian structure is given by
\be
h = \big(1-\tfrac12\,(\partial_y\psi)^2\big) \, \de x\otimes \de x + \partial_x\psi \, \partial_y\psi \, \de x \odot \de y + \big(1-\tfrac12\,(\partial_x\psi)^2\big) \, \de y \otimes\de y \ .
\ee

The torsion of an Aristotelian fluid flowing on a two-dimensional Euclidean domain is given by the two-form
\be
\bar T^{\nabla^\ttA} = -\bigg[\Big(\Delta_\mu\psi + \partial_x\psi \, \frac\partial{\partial x} + \partial_y\psi \, \frac\partial{\partial y}\Big) \, \frac1{(\partial_x\psi)^2 + (\partial_y\psi)^2}\bigg] \ \de x \wedge \de y \ .
\ee
\end{example}

\begin{example}
We consider a simple classical example. Consider the stream function
\be \label{eq:streamcd}
\psi(t,x,y) = \tfrac12\,\big[A(t)\,x^2 + B(t)\,y^2\big] \ ,
\ee
which is discussed by~\cite{Napper:2023jky} in connection with the occurence of metric singularities of the Monge-Amp\`ere geometry of the fluid flow associated to vanishing vorticity.
Here we take $A$ and $B$ to be non-zero functions of time $t\in\IR$ alone, and  restrict the domain of $\psi$ to be the simply connected open region $M^2=\{(x,y)\in\IR^2\ |\ x,y>0\}$ where it has no critical points. 

The fluid flows with uniform vorticity given by
\be
\xi = A+B \ .
\ee
The vorticity equation \eqref{eq:vorteq2d} implies that the flow is steady, i.e. $A+B$ is conserved. For the special case where $A$ and $B$ are each separately conserved, the flow is also stationary with Bernoulli function
\be
b(x,y) = \tfrac12\,(A+B)\,(A\,x^2 + B\,y^2) \ ,
\ee
up to an additive constant which can be fixed by specifying $b|_{\partial M^2}$.

The leaves of the one-parameter family of foliations $\cF$ are given by the streamlines
\be
L_r(t) = \big\{(x,y)\in M^2\ \big|\ A(t)\,x^2 + B(t)\,y^2 = r\big\} \ ,
\ee
for $r\in\IR$. For $r\neq0$, this foliates the fluid domain $M^2$ by quarter-ellipses or quarter-hyperbolas depending on the relative signs of the parameters $A$, $B$ and $r$. 

The torsion is given by the two-form
\be
\bar T^{\nabla^\ttA} = \frac{A+B}{(A\,x)^2+(B\,y)^2} \ \de x\wedge\de y \ .
\ee
A torsion-free Aristotelian fluid flow is thus only possible when the vorticity vanishes, i.e.~$A=-B$. Then the streamlines are unbounded and $M^2$ is foliated by rectangular quarter-hyperbolas for $r\neq0$. In this example, the torsion-free regions  coincide with the singular regions observed in~\cite{Napper:2023jky} where the Monge-Amp\`ere metric is Kleinian. In this sense, torsion is a desirable feature of an Aristotelian fluid.
\end{example}

\begin{example}
Let us now look at a simple classical example of an unsteady flow. Consider the stream function
\be
\psi(t,x,y) = -x^2 + 3\,y\,t + y^3 \ ,
\ee
which is also discussed by~\cite{Napper:2023jky} in connection with the occurence of scalar curvature singularities of the Monge-Amp\`ere geometry associated to topological bifurcations in the fluid flow. The fluid domain is $M^2=\IR^2$, and here we restrict to flows in a time parameter $t>0$, so that $\psi$ has no critical points.

The vorticity is the function
\be
\xi = 2\,(3\,y-1) \ .
\ee
The vorticity equation \eqref{eq:vorteq2d} reads
\be
\frac{\partial\xi}{\partial t} = -12\,x \ ,
\ee
and is simply a consequence of the flow equations for the fluid lines.

The torsion is given by the two-form
\be
\bar T^{\nabla^\ttA} = -(3\,y+1) \, \frac{8\,x^2 - 18\,(t+y^2)^2}{\big(4\,x^2+9\,(t+y^2)^2\big)^2} \ \de x\wedge \de y \ .
\ee
For each $t>0$, the Aristotelian fluid flow is torsion-free on the line $y=-\frac13$ and along the parabolas in $\IR^2$ defined by
\be
x = \pm\,\tfrac32\, (t+y^2) \ .
\ee
The torsion-free parabolas change with time, but they always contains points with $y=\frac13$ at which the vorticity vanishes and where the scalar curvature of the Monge-Amp\`ere metric is singular~\cite{Napper:2023jky}. On the other hand, there are torsion-free points with non-zero vorticity, as well as points in $\IR^2$ with vanishing vorticity but non-zero torsion.
\end{example}

\medskip

\subsection{The Godbillon-Vey Class of an Aristotelian Fluid} ~\\[5pt]
Since all the properties of our fluid determined by an integrable Aristotelian structure have now been established, let us turn to its characteristic Godbillon-Vey class in dimensions $n\geq 3$.
The condition \eqref{eq:ivTAnn} on the torsion tensor guarantees that the Godbillon-Vey class of the fluid is  non-trivial in general. In particular, a non-trivial Godbillon-Vey class is a first order obstruction to an $\sfS\sfO(n-1)$-frame moving along the fluid flow (see Remarks~\ref{rmk:integrGAri} and~\ref{rmk:inttorsionArifluid}). It also obstructs the fluid velocity $v$ from being a Killing vector field of the background Riemannian metric $g$ (see Remark~\ref{rmk:inttorsionArifluid}). By Remark~\ref{rmk:GVtauv}, $\gv(\tau) = \gv(v^\flat).$

The fluid speed $s_v$ is constant along the fluid lines when the background Riemmanian metric for the fluid flow is characterised solely by the Aristotelian structure, i.e. $s_v=1$ in Equation \eqref{eq:metricfluid}. We will see in the following that such fluid flows are robustly characterised by their Godbillon-Vey class. 

\begin{proposition} \label{prop:gvsteady}
The Godbillon-Vey class $\gv(\tau)=\gv(v^\flat)=[\alpha_v \wedge \de \alpha_v]$ for an Aristotelian fluid can be expressed as
\be \label{eq:GVfluid}
\alpha_v \wedge \de \alpha_v  = - \frac{1}{s_v^2} \, \de \bigl((\pounds_v s_v) \, \omega_\xi \bigr) + \frac{1}{s_v^2} \, (\iota_v\, \omega_\xi) \wedge \Big(\frac{\del \xi}{\del t} \,\lrcorner\, \mu\Big) \ .
\ee
If the fluid speed $s_v$ is constant along the fluid lines, then 
\be \label{eq:GVfluidsteady}
\alpha_v \wedge \de \alpha_v  = \frac{1}{s_v^2}\, (\iota_v\, \omega_\xi ) \wedge \Big(\frac{\del \xi}{\del t} \,\lrcorner\, \mu\Big) \ ,
\ee
and hence the Godbillon-Vey class is an obstruction to a steady fluid flow.
\end{proposition}

\begin{proof}
From Equations \eqref{eq:alphafluid} and \eqref{eq:Zalpha} we obtain
\begin{align}
(\pounds_v\xi) \,\lrcorner\, \mu &=  \pounds_v\, \omega_\xi   \\[4pt]
&= \de \bigl( s_v\,  \alpha_v +  (\pounds_v \log s_v) \, v^\flat  \bigr)   \\[4pt]
&= \de s_v \wedge \alpha_v + s_v\, \de \alpha_v + \de (\pounds_v \log s_v) \wedge v^\flat + (\pounds_v \log s_v)\, v^\flat \wedge \alpha_v \ .
\end{align}
By taking the exterior product with $\alpha_v$ and using the vorticity equation \eqref{eq:vorticity} together with Equation~\eqref{eq:alphafluid} we find
\be \label{eq:gvfluid1}
\begin{split}
s_v\, \alpha_v \wedge \de \alpha_v &= \frac{1}{s_v}\, \de (\pounds_v \log s_v) \wedge \, (\iota_v\, \omega_\xi) \wedge v^\flat \\
& \qquad \, - \frac{1}{s_v}\, \big( (\pounds_v \log s_v)\, v^\flat - \iota_v\, \omega_\xi\big) \wedge \Big(\frac{\del\xi}{\del t} \,\lrcorner\, \mu\Big) \ .
\end{split}
\ee
We further observe that
\begin{align}\label{eq:gvfluid2}
\begin{split}
\de (\pounds_v \log s_v) \wedge \, (\iota_v\, \omega_\xi) \wedge v^\flat &= \de \bigl( (\pounds_v \log s_v) \, (\iota_v\, \omega_\xi) \wedge v^\flat  \bigr) \\
& \quad \, + (\pounds_v \log s_v) \, \bigg(\Big(\frac{\del\xi}{\del t} \,\lrcorner\, \mu\Big) \wedge v^\flat +  (\iota_v\, \omega_\xi) \wedge \omega_\xi \bigg)   \ .
 \end{split}
\end{align} 
By combining Equations \eqref{eq:gvfluid1} and  \eqref{eq:gvfluid2} using
\be  
(\iota_v\, \omega_\xi) \wedge \omega_\xi = \iota_v\, (\xi\,\lrcorner\, \mu) \wedge (\xi\,\lrcorner\, \mu) = 0 
\ee
along with
\be  
\tfrac{1}{s_v}\, (\iota_v\, \omega_\xi) \wedge v^\flat = \alpha_v \wedge v^\flat = - \omega_\xi \ ,
\ee
we obtain Equation~\eqref{eq:GVfluid}.

The expression \eqref{eq:GVfluidsteady} is straightforwardly obtained by imposing the condition \eqref{eq:transvelocity}. It is also clear that 
\be  
\frac{\del\xi}{\del t} = 0
\ee
is possible only if the Godbillon-Vey class of the fluid is trivial.
\end{proof}

Proposition~\ref{prop:gvsteady} demonstrates that, since $\gv(\tau)=\gv(v^\flat),$ the spacetime structure determines whether a steady flow is possible to realise. In particular, it follows from Remarks~\ref{rmk:integrGAri} and~\ref{rmk:inttorsionArifluid} that non-integrability of the underlying $\sfS \sfO(n-1)$-structure can present an obstruction to the existence of steady solutions of the Euler equations. 
Equation \eqref{eq:GVfluidsteady} moreover gives information about the defining component of the torsion tensor: It follows from Equation~\eqref{eq:alpT} that \smash{$\iota_v\bar{T}^{\nabla^\ttA}$} is determined by the vorticity $\omega_\xi$ up to exact terms.    

\begin{remark} \label{rmk:results17}
For a steady flow with Bernoulli function $b\in C^\infty(M^n)$ (see Remark~\ref{rmk:eulerisable}), if the fluid speed $s_v$ is constant along the fluid lines, then by solving Equation~\eqref{eq:arieulerisable} we find
\be  
\alpha_v = \tfrac{1}{s_v}\, \de b \ ,
\ee
and thus the Godbillon-Vey class of the fluid is trivial.

Proposition \ref{prop:gvsteady} in the case $\pounds_v s_v = 0$ gives a condition for the existence of a steady flow with a Bernoulli function for an Aristotelian fluid. Moreover, the construction of a steady flow above is consistent with this condition, i.e. it always yields a fluid with trivial Godbillon-Vey class if the condition \eqref{eq:alphatransv} is satisfied.
\end{remark}

\begin{example}[\textbf{Fluid Flows on Warped Products}]
 Let us look at an example which can be regarded as a local model for any Aristotelian fluid flow. We consider a fluid domain which is a direct product $M^n = M^{n-1} \times N$ for $n\geq3$. Here $M^{n-1}$ is an orientable $n{-}1$--dimensional manifold endowed with a Riemannian metric $h$ as well as a volume form $\mu_{n-1}$ and a smooth function $\varphi \in C^\infty(M^{n-1})$, while $N$ is a one-dimensional manifold admitting a smooth function $q \in C^\infty(N)$ with no critical points. The Riemannian metric $g$ on $M^n$ is taken to be of warped product form
 \begin{align}
     g = h + \tfrac{1}{2}  \e^\varphi \, \de q \otimes \de q  \ ,
 \end{align}
 while the volume form $\mu$ on $M^n$ is
 \be
 \mu = \mu_{n-1} \wedge \de q \ .
 \ee
 
 Let us consider the one-parameter family of nowhere-vanishing one-forms
 \begin{align}
     \tau =  F(t) \, \de q \ \in \ \mathsf{\Omega}^1(N) \ ,
 \end{align}
 where $F$ is a function of the time parameter alone such that $F(t),F'(t)\neq0$ for all $t\in\IR$ and $F(t_0)=1$ for some initial time $t=t_0$. For the fluid velocity we take
 \begin{align}
     v = \frac{v_q}{F(t)} \ ,
 \end{align}
 where $v_q \in \mathsf{\Gamma}(TN)$ is the vector field such that $\iota_{v_q} \, \de q = 1 .$ Then the fluid lines run along the one-dimensional manifold $N$. The vorticity covector potential is given by
 \begin{align} \label{eq:vflatwarped}
     v^\flat = \frac{\e^{\varphi}}{F(t)} \, \de q = \frac{\e^{\varphi}}{F(t)^2} \, \tau \ ,
 \end{align}
 and so the speed of the fluid is  
\be
s_v = \frac{\e^{\varphi}}{F(t)^2} \ . 
\ee
Since $\iota_{v_q}\,\de\varphi=0$, the speed is constant along the fluid lines, i.e. $\pounds_vs_v=0$.

The data $(\mu, \tau, g,v,p)$ define an Aristotelian fluid flowing in the domain $M^n$, such that $\ker(\tau) = TM^{n-1}$ and $\de \tau = 0$, which corresponds to the foliation $\cF$ whose leaves are the fibres $M^{n-1}$ of the trivial bundle $M^n$ over $N$. In particular, the torsion of the Aristotelian structure vanishes and its Godbillon-Vey class is trivial. The pressure field $p$ will be discussed below.

 We can compute the vorticity starting from 
 \begin{align}
     \de v^\flat = \de \varphi \wedge v^\flat \ ,
 \end{align}
 which is easily obtained from Equation \eqref{eq:vflatwarped}. This identifies the one-form $\alpha_v$ tangential to the foliation $\cF$ as
 \be
\alpha_v = -\de\varphi \ \in \ \mathsf{\Omega}^1(M^{n-1}) = \mathsf{\Gamma}(\mathrm{Ann}(v)) \ .
 \ee
 This also identifies the vorticity two-form  as
 \begin{align}
     \omega_\xi = \frac{\e^{\varphi}}{F(t)} \, \de \varphi \wedge \de q \ .
 \end{align}
 
 We decompose the vorticity $n{-}2$-vector $\xi \in \mathsf{\Gamma}(\midwedge^{n-2}\,TM^n)$ as
 \be
 \xi = \xi_{n-2} + \xi^i_{n-3}\wedge\zeta_i \ ,
 \ee
 where $\xi_{n-2} \in \mathsf{\Gamma}(\midwedge^{n-2}\,TM^{n-1})$, $\xi^i_{n-3}\in \mathsf{\Gamma}(\midwedge^{n-3}\,TM^{n-1})$ and $\zeta_i \in \mathsf{\Gamma}(TN)$. This gives
 \begin{align} 
 \xi \,\lrcorner\, \mu = (\xi_{n-2} \, \lrcorner \, \mu_{n-1}) \wedge \de q + (\iota_{\zeta_i}\, \de q) \  \xi^i_{n-3} \, \lrcorner \, \mu_{n-1} = \omega_\xi \ ,
 \end{align}
 which implies that either $\zeta_i=0$ or $\xi^i_{n-3}=0$ for each $i$. Thus
 \begin{align}
     \xi=\xi_{n-2} = \frac{\e^{\varphi}}{F(t)} \, \star_{\mu_{n-1}}  \de \varphi \ .
 \end{align}
 
Lemma \ref{lemma:arifluid} now allows us to write the flow equation for $v^\flat$ as an equation for the pressure field:
 \begin{align} \label{eq:pressureeq}
     \de p = \frac1{F(t)^2} \, \big( F'(t) \, \e^{\varphi} \, \de q - \tfrac{3}{2} \, \de \e^\varphi \big) \ .
 \end{align}
 Since $q \in C^\infty(N)$ and $\varphi \in C^\infty(M^{n-1}),$ the left-hand side  decomposes respectively into two pieces as $\de p =\de p_\tau + \de p_v$, with $\de p_\tau\in\mathsf{\Gamma}(\mathrm{Span}(\tau)) = \mathsf{\Omega}^1(N)$ and $\de p_v\in\mathsf{\Gamma}(\mathrm{Ann}(v))=\mathsf{\Omega}^1(M^{n-1})$. Since the Godbillon-Vey class of the fluid is trivial, there is no obstruction to stationary solutions of the Euler equations: a steady flow merely requires $F'(t)=0$.
 
 In particular, for a stationary flow, since $v=v_q$ at the initial time $t=t_0$, by Remark~\ref{rmk:results17} it follows that the Bernoulli function is given by the warp factor
 \be
b = -\e^\varphi \ ,
 \ee
 up to additive constant functions. This fixes $p_\tau=0$ and the function $p_v \in C^\infty(M^{n-1})$ up to additive constant functions as
 \begin{align}
  p_v = -\tfrac{3}{2} \e^{\varphi} \ ,
 \end{align}
 which also follows directly from Equation~\eqref{eq:pressureeq}.
 \end{example}

Let us finally look at the flow equation for the Godbillon-Vey class of the Aristotelian structure, which is given by

\begin{proposition} \label{prop:evolgv}
Let $(M^n, \mu,\tau, v , g,p)$ be an Aristotelian fluid. Then the flow equation for the Godbillon-Vey class $\gv(\cF) = [\alpha_v \wedge \de \alpha_v]$ of the fluid is given by
\begin{align} \label{eq:evogvc}
\begin{split}
\hspace{-1cm} \frac\De{\De t}\big( \alpha_v \wedge \de \alpha_v\big) &= \de \Bigl( \big( \kappa \, v^\flat + \tfrac{(-1)^n}{s_v}\, (\pounds_v\log s_v) \, \iota_v\, \omega_\xi -\tfrac1{s_v}\,(\pounds_v \log s_v) \, \de (p- \tfrac{s_v}{2})  \big) \wedge \alpha_v \Bigr) \ . 
\end{split}
\end{align}
\end{proposition}

\begin{proof}
The respective Leibniz rules yield
\begin{align} \label{eq:evgv1}
\begin{split}
\frac\De{\De t}\big( \alpha_v \wedge \de \alpha_v\big)  &= \frac{\De \alpha_v}{\De t} \wedge \de \alpha_v  + \alpha_v \wedge \frac\De{\De t} \, \de \alpha_v  = - \de\Big(\frac{\De\alpha_v}{\De t}\wedge\alpha_v \Big) \ ,
\end{split}
\end{align}
where we used the fact that the exterior derivative commutes with both the time derivative and the Lie derivative. Substituting Equation~\eqref{eq:evalpha} into Equation~\eqref{eq:evgv1} gives Equation~\eqref{eq:evogvc}.
\end{proof}

Proposition~\ref{prop:evolgv} shows that the Godbillon-Vey class $\gv(\cF)$ is transported exactly by the fluid flow.

\medskip

\subsection{Three-Dimensional Aristotelian Fluid Flows} ~\\[5pt]
We conclude by focusing on the special geometric properties  exhibited by Aristotelian fluids in three dimensions, and present some concrete examples. 
Recall from Example~\ref{rmk:vort} that vorticity $\xi=\mathrm{curl}(v)$ is a vector field in three dimensions, and the integrability condition for the  covector potential $v^\flat$ implies that the helicity of the vorticity vanishes, thus $\xi$ is tangent to the leaves of the foliation $\cF$ of the integrable Aristotelian structure, i.e. $\xi \in \mathsf{\Gamma}(T\cF).$ Unlike planar flows, this severely restricts the possible three-dimensional incompressible flows. We will see below that the Godbillon-Vey class in this instance generally provides a non-vanishing higher order topological invariant of the fluid flow which is a conserved quantity. From this perspective, torsion is once again a desired property of an Aristotelian fluid flow.

The special properties of fluid flows for $n=3$ whose speed $s_v$ is constant along the fluid lines, observed in Remark~\ref{rmk:arimetric}, have previously appeared in the literature. For instance, the one-form $\alpha_v$ considered in \cite{Machon2020} (denoted $\eta$ in that paper) satisfies the condition \eqref{eq:alphatransv}.
Applying~\cite[Proposition~1.2]{Rovenski2019} to our fluids, in this case a non-trivial Godbillon-Vey class obstructs the velocity $v$ from being a geodesic for the background Riemannian metric $g$. 

In three dimensions, Eulerisable vector fields (see Remark~\ref{rmk:eulerisable}) are those divergence-free vector fields $v$ for which
\be
v\times \xi = -(\de b)^\sharp \ ,
\ee
where $v\times\xi := -(\iota_v\,\iota_\xi\mu)^\sharp$.
It follows that Eulerisable Aristotelian fluid flows in three dimensions cannot accommodate locally constant Bernoulli functions $b$, i.e. the \emph{Bernoulli fields} $v$ which are parallel to their vorticity $\xi=\mathrm{curl}(v)$. This excludes some classic examples of steady solutions to the Euler equations with non-zero helicity, such as Hopf fields on $S^3$ 
and ABC flows on $T^3.$ For a wide class of Eulerisable flows on three-manifolds with non-constant Bernoulli function see \cite{Cardona2022}.

Specialising Remark~\ref{rmk:results17} to three dimensions, a more general statement can be made from results of~\cite{Rechtman:2020}: ideal fluids with Eulerisable flow on a three-dimensional manifold $M^3$ with 
$\sfH^1(M^3;\IR)=0$ always have trivial Godbillon-Vey class.
The case $s_v=1$ corresponds to the metric constructed in the proof of~\cite[Theorem~1.4]{Rechtman:2020}. 
Proposition \ref{prop:gvsteady} is consistent with the results of \cite{Rechtman:2020} discussed in Remark~\ref{rmk:results17}, since the Godbillon-Vey class of an Eulerisable flow must be trivial in order to allow for the existence of a stationary solution to the Euler equations, i.e. a steady flow.

Let us take a closer look at Equation \eqref{eq:GVfluidsteady} for three-dimensional fluid flows with trivial Godbillon-Vey class. We consider the case in which the fluid flow is unsteady.
Thus
\be  
0 = \iota_v\, \omega_\xi \wedge \iota_{\frac{\del\xi}{\del t} } \mu =  \iota_v\,\iota_\xi \mu \wedge \iota_{[v, \xi]} \mu \ ,
\ee
which implies that
\be  
[v, \xi] =  \phi \, \xi \ ,
\ee
for some function $\phi \in C^\infty(M^3).$ Together with the vorticity equation $\frac{\partial\xi}{\del t}=[\xi,v]$ this means that, for an unsteady fluid flow with trivial Godbillon-Vey class, the vorticity vector field evolves in time whilst preserving its direction. In other words, the vorticity remains orthogonal to the velocity vector field, i.e. tangent to the leaves of the foliation $\cF,$ for all times.
This is a necessary condition for Aristotelian fluids whose density $\mu$ is spatially constant.

When the speed is constant along the fluid lines, i.e. $\pounds_v s_v = 0$, the result of Proposition~\ref{prop:evolgv} for $n=3$ specialises to 

\begin{corollary} \label{cor:evogvrep}
Let $(M^3,\mu,\tau,v,g,p)$ be a three-dimensional Aristotelian fluid. Suppose that the speed of the fluid $s_v$ is constant along the fluid lines. Then the flow equation for the Godibillon-Vey class of the fluid $\gv(\cF) = [\alpha_v \wedge \de \alpha_v], $ where $\de v^\flat = v^\flat \wedge \alpha_v$, is given by
\be \label{eq:evogvs}
\frac\De{\De t} (\alpha_v \wedge \de \alpha_v) = \de (\kappa \, \omega_\xi ) \ ,
\ee 
where $\omega_\xi = \iota_\xi\mu$ is the vorticity two-form and 
\be  
\kappa = - g^{-1}(\alpha_v, \alpha_v) - \tfrac1{s_v}\,g^{-1}\big(\de (p + \tfrac{s_v}{2}), \alpha_v\big) - \tfrac1{s_v^2} \, g^{-1}\big(\iota_v\, \omega_\xi + \de (p+\tfrac{s_v}2) , \de s_v \big)  - \tfrac2{s_v^2}\,\pounds^2_v\,p  \ .
\ee
\end{corollary}

\begin{proof}
This follows straightforwardly from Equation~\eqref{eq:evogvc} by invoking $\pounds_v s_v = 0,$  and noticing that
\be  
\omega_\xi = v^\flat \wedge \alpha_v \ 
\ee
in this case.
\end{proof}

Proposition \ref{prop:evolgv} and Corollary~\ref{cor:evogvrep} show how the geometric structure of an Aristotelian fluid in three dimensions gives rise to new explicit conservation laws. 
For $n=3$, our flow equation \eqref{eq:evogvc} generalises the local conservation law from \cite{Machon2020}. In particular, Equation~\eqref{eq:evogvs} is analogous to the transport equation obtained in \cite[Section~5.1]{Machon2020}, and in this case the Godbillon-Vey class of the fluid is carried by its vorticity.

\begin{definition}
The \emph{Godbillon-Vey number} $\gvs(\cF)$ of the foliation $\cF$ is the integral\footnote{Here and throughout the rest of the section we assume that integration on $M^3$ is well-defined. For simplicity, one may take $M^3$ to be compact.} invariant associated to the Godbillon-Vey class $\gv(\cF)=[\alpha_v\wedge\de\alpha_v]$:
\be  
\gvs(\cF) \coloneqq \int_{M^3} \, \alpha_v \wedge \de \alpha_v \ .
\ee
\end{definition}

We can establish how a global conservation law arises from the Godbillon-Vey invariant for a three-dimensional Aristotelian fluid flow through

\begin{proposition}
The Godbillon-Vey number $\gvs(\cF)$ of an Aristotelian fluid in three dimensions is a conserved quantity along the fluid lines if $M^3$ is closed, i.e. it yields the conservation law
\be \label{eq:consgv}
\frac{\de\,\gvs(\cF)}{\de t} = 0 \ .
\ee
If $\del M^3 \neq \varnothing,$ then Equation~\eqref{eq:consgv} holds if the vorticity vector field $\xi = \mathrm{curl}(v)$ is parallel to the boundary of $M^3$. 
\end{proposition}
\begin{proof}
Since the velocity $v\in \mathsf{\Gamma}_\mu(TM^3)$ preserves the volume form $\mu,$ we can integrate Equation~\eqref{eq:evogvc} to obtain
\begin{align*}
\frac{\de\,\gvs(\cF)}{\de t} &= \int_{M^3} \, \frac\De{\De t}\big(\alpha_v \wedge \de \alpha_v\big) \\[4pt]
&=  \int_{M^3} \, \de \Bigl( \big( \kappa \, v^\flat + \tfrac{1}{s_v}\, (\pounds_v\log s_v) \, \iota_{v}\, \omega_\xi -\tfrac1{s_v}\,(\pounds_v \log s_v) \, \de (p- \tfrac{s_v}{2}) \big) \wedge \alpha_v \Bigr) \ .
\end{align*}
By Stokes' Theorem the right-hand side is a boundary integral, which vanishes because both $v$ and $\xi$ are taken to be parallel to the boundary $\partial M^3$. By the vorticity equation $\frac{\del\xi}{\del t} = [\xi,v]$, the flow of the vorticity vector field ensures that $\xi$ remains parallel to the boundary at any time. When $M^3$ is closed this vanishes without any further conditions.
\end{proof}

\begin{remark}
The conservation law arising from the stronger setting of Corollary~\ref{cor:evogvrep}, for a fluid flow on a three-manifold with non-empty  boundary, relies solely on the assumption that the vorticity vector field $\xi$ is parallel to the boundary and hence so is its flow.  
\end{remark}

The Godbillon-Vey number $\gvs(\cF)$ is the helicity $\cH(\zeta)$ of the null homologous vector field $\zeta$ defined by~\cite{Arnold}
\be  
\iota_\zeta \mu = \pounds_v(\pounds_v \tau) \wedge \tau \ .
\ee
The vector field $\zeta$ measures the angular acceleration of the rotation determined by the vorticity $\xi.$ This can be potentially extended to the Carrollian hydrodynamics of~\cite{Freidel2022} as a relevant case in which the helicity of the vorticity is non-vanishing.

\begin{remark}
The helical wobble discussed in Example \ref{rmk:helical} inspires the following local geometric interpretation of the Godbillon-Vey class as helical compression of vorticity for an Aristotelian fluid in three dimensions, which is reflected in the non-linearity of the Euler equations. Our interpretation builds on and extends the discussion of~\cite{Machon2020} (see also~\cite{Webb2019}). 

Let $\alpha \in \mathsf{\Omega}^1(M^3)$ be a one-form satisfying the integrability equation for the clock form $\tau.$ Then the dual vector field $\alpha^\sharp$ measures the local compression (or expansion) of the leaves of the foliation $\cF,$ i.e. the sections of the fluid flowing along the fluid lines. In particular, its norm $g(\alpha^\sharp, \alpha^\sharp)$ measures the curvature of the fluid lines, which are normal to $\cF$ with respect to $g.$ This quantity also determines the flow of the Godbillon-Vey class $\gv(\cF)$, as seen in Equation \eqref{eq:evogvs} where it is the main contributing factor together with the vorticity $\xi$ and density $\mu,$ if the pressure $p$ is constant and the speed of the fluid $s_v$ is constant along the fluid lines, i.e. $\pounds_v s_v = 0.$

The direction of $\alpha^\sharp$ determines the direction in which the leaves of $\cF$ expand. Because
$g(\alpha^\sharp, \mathrm{curl}(\alpha^\sharp)) \neq 0,$ the twist of $\alpha^\sharp$ transverse to the leaves measures the topological helical compression of the vortex lines. The Godbillon-Vey class $\gv(\cF)$ measures the local spin of $\alpha^\sharp$ in the direction of the fluid lines, i.e. transversally to the fluid sections determined by $\cF.$ The Godbillon-Vey number $\gvs(\cF)$ gives a global measure of this spin.
\end{remark}

\begin{example}[\textbf{Hydrodynamics with Roussarie Foliations}]
We study Aristotelian fluid flows, with non-trivial torsion and Godbillon-Vey invariant, on three-dimensional domains that admit a \emph{Roussarie foliation}, which we construct following \cite{Morita}. Consider the Lie group $\mathsf{PSL}(2, \IR)$ whose Lie algebra
${\mathfrak{sl}}(2, \IR)$ is characterised by a basis of generators  $\mathsf{b} = \{T_0 , T_1, T_2\}$ in which the Lie brackets are
\begin{align} \label{liealgebrapsl}
    [T_1, T_2] = T_0 \ , \quad [T_0, T_1] = 2\, T_1 \qquad \mbox{and} \qquad [T_0, T_2] = - 2\, T_2 \ . 
\end{align}
Let $\mathsf{\Lambda} \subset \mathsf{PSL}(2, \IR)$ be a torsion-free cocompact discrete subgroup, acting on $\mathsf{PSL}(2, \IR)$ by left multiplication. Then $M^3 \coloneqq \mathsf{\Lambda} \setminus \mathsf{PSL}(2, \IR)$ is a compact connected three-manifold. 

The three-manifold $M^3$ inherits the global frame $\mathsf{F}= \{ X_0, X_1, X_2 \}\subset \mathsf{\Gamma}(TM^3)$ from the left-invariant vector fields on $\mathsf{PSL}(2, \IR)$ associated with the basis of generators $\mathsf{b}.$ The Lie subalgebra 
\begin{align}
    \mathfrak{l} \coloneqq \mathrm{Span}\{ T_0, T_1 \} \ \subset \ \mathfrak{sl}(2, \IR) 
\end{align}
induces a left-invariant foliation of $\mathsf{PSL}(2, \IR)$ that descends to $M^3$, which is characterised as follows. Let $\mathsf{F}^{\textrm{\tiny$\vee$}} = \{ \theta_0, \theta_1, \theta_2 \} \subset \mathsf{\Omega}^1(M^3) $ be the global coframe dual to $\mathsf{F}.$ Then the foliation $\cF$ of $M^3$ is given by the distribution
$T\cF = \ker (\theta_2),$ where the one-form $\theta_2$ satisfies the integrability condition \eqref{eq:alphaint} with
\begin{align} \label{eq:introussarie}
 \de \theta_2  = 2\, \theta_2 \wedge \theta_0 \ , 
\end{align}
which is a consequence of the Maurer-Cartan equations for the Lie group $\mathsf{PSL}(2, \IR)$. 

We define an Aristotelian fluid flowing in $M^3$ by taking, at an initial time $t=t_0$, the Aristotelian structure on $M^3$ to be given by the clock form and velocity vector field
\begin{align}
\tau=\theta_2 \qquad \mbox{and} \qquad v=X_2 \ ,
\end{align}
together with any Riemannian metric $h$ on $T\cF.$ A natural choice for the spatial metric $h$ is the descendant from the left-invariant metric on the foliation of $\mathsf{PSL}(2, \IR)$ for which its left-invariant generating vector fields are orthonormal:
\be\label{eq:hRoussarie}
h =\theta_0 \otimes \theta_0 + \theta_1 \otimes \theta_1 \ .
\ee

Then we allow the Aristotelian structure to evolve in time according to the Euler equations on $M^3$, with respect to the background Riemannian metric 
\begin{align}
    g = h + \tfrac12\,\theta_2 \otimes \theta_2 \ .
\end{align}
Hence the initial speed of the fluid is $s_v=1.$ Assuming that the speed remains constant at all times $t\in\IR$, by Proposition~\ref{prop:transpf} it follows that only pressure fields $p \in C^\infty(M^3)$ that are invariant along the fluid lines are permissible, i.e. $\pounds_{X_2} p = 0.$ 

Lastly, the fluid density $\mu$ is given by the volume form induced on $M^3$ by the natural left-invariant volume form on $\mathsf{PSL}(2, \IR)$: 
\be
\mu = \theta_0 \wedge \theta_1 \wedge \theta_2 \ .
\ee
Hence the condition $\pounds_{X_2} \mu = 0$ for divergence-free flow is satisfied. Equation \eqref{eq:introussarie} yields  
\begin{align}
    \alpha= \alpha_v = 2\, \theta_0 \ .
\end{align}
Together with the fluid density $\mu$, this determines the vorticity vector field and two-form
\begin{align}
    \xi = 2\, X_1 \qquad \mbox{and} \qquad \omega_\xi = 2\,\theta_2\wedge\theta_0 \ .
\end{align}

This provides a complete characterisation of the fluid at the initial time $t=t_0$. Let us now determine its flow equations.
By Lemma~\ref{lemma:arifluid} the Euler flow equation reads
\begin{align} \label{eq:flowroussarie}
 \frac{\del X_2}{\del t} = (- 2\, \theta_0 -\de p)^\sharp \ .   
\end{align}
With respect to the choice of spatial metric $h$ in Equation~\eqref{eq:hRoussarie}, the dual of the flow equation \eqref{eq:flowroussarie} gives the transport equation for the clock form~$\tau=\theta_2$:
\begin{align} \label{eq:theta2}
 \frac{\del \theta_2}{\del t} + 2 \, \theta_0 = -\de p \ .   
\end{align}

Using Equation~\eqref{eq:torsionAri}, in this case the torsion tensor is determined as the two-form
\be
\bar{T}^{\nabla^\ttA} = 2\,\theta_2\wedge \theta_0 \ .
\ee
Since $s_v=1$ and $\pounds_{X_2}p = 0,$ from Equation~\eqref{eq:transpTors} it follows that it gives rise to a conservation law
\begin{align}
\frac{\De \bar{T}^{\nabla^\ttA}}{\De t} = \frac{\del \bar{T}^{\nabla^\ttA}}{\del t} = 0 \ .    
\end{align}

Foliations such as the Roussarie foliation $\cF$ of the quotient manifold $M^3$ are well-known to have a non-trivial Godbillon-Vey class $\gv(\cF)$. Here it is represented by the three-form
\begin{align}
\alpha_v \wedge \de \alpha_v=  \alpha \wedge \de \alpha = 4 \, \theta_0 \wedge  \theta_1 \wedge \theta_2  \ ,  
\end{align}
which by Corollary~\ref{cor:evogvrep} has flow equation
\begin{align} \label{eq:GVflowRoussarie}
 \frac\De{\De t}\big( \alpha_v \wedge \de \alpha_v\big) = 4\,\frac{\del}{\del t} \big( \theta_0 \wedge \theta_1 \wedge \theta_2 \big)= \de \kappa \wedge \bar{T}^{\nabla^\ttA} \ ,
\end{align}
where 
\be
\kappa = -2\,  g^{-1}(2 \, \theta_0 + \de p ,\theta_0) \ .
\ee 
Hence the torsion of the flow determines the time evolution of the Godbillon-Vey class, together with the pressure field $p$ and the spatial metric $h$. If the spatial metric is given by Equation~\eqref{eq:hRoussarie}, then $\kappa = -2 \, ( 2 + \iota_{X_0}\, \de p).$ 

Since the torsion tensor is a conserved quantity, Equation~\eqref{eq:GVflowRoussarie} simplifies to
\begin{align} \label{eq:theta1}
 2 \, \frac{\del \theta_1} {\del t} = \de \kappa \ ,   
\end{align}
which determines the time evolution of $\theta_1.$ From Equation~\ref{eq:evalpha} the flow equation for the one-form $\alpha_v$ gives
\begin{align}
 \frac{\del \theta_0}{\del t} + \theta_1 = - \frac{\kappa}{2} \, \theta_2  \ .
\end{align}
 Together with Equations \eqref{eq:theta1} and \eqref{eq:theta2}, this completely determine the time evolution of the coframe $\mathsf{F}^{\textrm{\tiny$\vee$}} = \{ \theta_0, \theta_1, \theta_2\}$, given the set of initial data considered above. Conversely, given the flow of the coframe, the pressure field can be determined.

By Proposition~\ref{prop:gvsteady}, in this example the flow equations cannot have a stationary solution. This has a direct consequence on the time dependence of the pressure field $p$.
\end{example}

\bibliographystyle{ourstyle}
\bibliography{biblioNC}

\end{document}